\documentclass[12pt]{article}
\pdfoutput=1
\usepackage{amsmath,amssymb,amsthm,appendix,bm,graphicx,float,array,multirow,multicol,rotfloat,caption,subcaption,hyperref,cleveref,enumerate,geometry,mathdots,adjustbox,booktabs,parskip,mathtools,tikz,tikz-cd,pdflscape,csquotes,lscape,rotating,empheq,braket,verbatim}
\usepackage{graphicx}
\usepackage{feynmp-auto}
\usepackage[para]{threeparttable}
\usepackage[all]{xy}
\usepackage[normalem]{ulem}
\usepackage[numbers,sort&compress]{natbib}
\numberwithin{equation}{section}
\numberwithin{figure}{section}
\allowdisplaybreaks
\makeatletter
\usetikzlibrary{arrows, positioning, decorations.pathmorphing, decorations.markings, decorations.pathreplacing, decorations.markings, matrix, patterns, cd}
\hypersetup{colorlinks=true}
\hypersetup{linkcolor=black}
\hypersetup{citecolor=black}
\hypersetup{urlcolor=black}
\makeatletter

\theoremstyle{plain}
\newtheorem*{thm*}{Theorem}
\newtheorem{thm}{Theorem}[section]

\newtheorem{prop}[thm]{Proposition}
\newtheorem{conj}[thm]{Conjecture}
\newtheorem{cor}[thm]{Corollary}

\theoremstyle{definition}
\newtheorem{defn}[thm]{Definition}
\newtheorem*{defn*}{Definition}
\newtheorem{exmp}[thm]{Example}

\newtheorem{rem}[thm]{Remark}

\crefname{lemma}{lemma}{lemmas}
\Crefname{lemma}{Lemma}{Lemmas}
\crefname{thm}{theorem}{theorems}
\Crefname{thm}{Theorem}{Theorems}
\crefname{defn}{definition}{definitions}
\Crefname{defn}{Definition}{Definitions}

\DeclarePairedDelimiterX{\abs}[1]{\lvert}{\rvert}{\ifblank{#1}{{}\cdot{}}{#1}}
\makeatother

\newtheorem*{thm:main}{Theorem \ref{thm:main}}
\newtheorem*{thm:prop}{Proposition \ref{thm:prop}}

\begin{document}

\begin{titlepage}
\vspace*{-3cm} 
\begin{center}
\vspace{2.2cm}
{\LARGE\bfseries Wormhole Renormalization:}\\\vspace{0.2cm}
 {\Large\bfseries The gravitational path integral, holography, and \\a gauge group for topology change}\\
\vspace{1cm}
{\large
Elliott Gesteau,$^{1}$ Matilde Marcolli,$^{1}$ and Jacob McNamara$^{2}$\\}
\vspace{.6cm}
{ $^1$ Division of Physics, Mathematics, and Astronomy, California Institute of Technology}\par\vspace{-.3cm}
{Pasadena, CA 91125, U.S.A.}\par
\vspace{.2cm}
{ $^2$ Walter Burke Institute for Theoretical Physics, California Institute of Technology}\par\vspace{-.3cm}
{Pasadena, CA 91125, U.S.A.}\par
\vspace{.4cm}

\scalebox{.95}{\tt  egesteau@caltech.edu, matilde@caltech.edu, jmcnamar@caltech.edu}\par
\vspace{1cm}
{\bf{Abstract}}\\
\end{center}
We study the Factorization Paradox from the bottom up by adapting methods from perturbative renormalization. Just as quantum field theories are plagued with loop divergences that need to be cancelled systematically by introducing counterterms, gravitational path integrals are plagued by wormhole contributions that spoil the factorization of the holographic dual. These wormholes must be cancelled by some stringy effects in a UV complete, holographic theory of quantum gravity. In a simple model of two-dimensional topological gravity, we outline a gravitational analog of the recursive BPHZ procedure in order to systematically introduce ``counter-wormholes" which parametrize the unknown stringy effects that lead to factorization. Underlying this procedure is a Hopf algebra of symmetries which is analogous to the Connes--Kreimer Hopf algebra underlying perturbative renormalization. The group dual to this Hopf algebra acts to reorganize contributions from spacetimes with distinct topology, and can be seen as a gauge group relating various equivalent ways of constructing a factorizing gravitational path integral.\\
\vfill 
\end{titlepage}

\tableofcontents
\newpage
\section{Introduction}

Should the path integral of quantum gravity include a sum over topologies? There are good reasons to think it must: quantum mechanics instructs us to sum the amplitudes for all allowed processes, and gravity is supposed to be a theory of dynamical spacetime manifolds of arbitrary topology. However, it has been known for quite some time that the sum over topologies leads to deep structural issues \cite{StHW, tHoo, Col, GS1, GS2}. These issues have been sharpened into the Factorization Paradox \cite{Witten:1999xp, Maldacena:2004rf}, a direct conflict between the holographic principle and a gravitational path integral over topologies. The holographic principle asserts that quantum gravity defines a local quantum field theory (QFT) living on the boundary of spacetime.\footnote{Strictly speaking, this is likely only exactly true in the context of AdS/CFT. While some much more general form of holography is expected to be true, the details are still far from clear.} However, the space of configurations in the gravitational path integral, given by the set of bulk spacetimes with fixed boundary, is not local to the boundary.\footnote{Mathematically, the Factorization Paradox is the fact that the configuration space in the gravitational path integral, given by the set of spacetime manifolds, does not define a sheaf over the space of boundaries.} In particular, while the holographic principle requires that correlation functions defined with a disconnected boundary must factorize, there are non-factorizing configurations in the gravitational path integral, otherwise known as spacetime wormholes.

There is essentially only one possible resolution to the Factorization Paradox: in a complete, holographic theory of quantum gravity, the net sum of all wormholes in the gravitational path integral must exactly vanish \cite{MaMax,MV,Saad:2021rcu,Saad:2021uzi,Blommaert:2021fob}.\footnote{A trivial way to achieve factorization would be give a description of the theory that manifestly factorizes. For example, we could attempt to excludes wormholes by fiat, as envisioned in \cite{StHW, tHoo}. Alternatively, we could retreat to the holographically dual description in terms of a local path integral on the boundary. The puzzle of the Factorization Paradox is to understand how a holographic theory could still admit a description in terms of a gravitational path integral which seems to include wormholes and does not manifestly factorize.} This cancelation of all wormholes has been demonstrated in a few very specialized corners of string theory, such as a particular tensionless limit \cite{Eberhardt:2021jvj} or when computing supersymmetric indices \cite{Iliesiu:2021are}. Importantly, these examples involve supersymmetric localization, and the necessary cancelations occur simply, between nearby wormhole configurations related by the action of supersymmetry. In more generic contexts, we should not expect any simple cancelation between wormholes beyond what is required by holography. Thus, even if we know the final answer must vanish, it is not  obviously clear how to organize the sum over wormholes in the absence of a complete theory of quantum gravity.

In \cite{Saad:2021rcu,Saad:2021uzi,Blommaert:2021fob}, a useful organizing principle has been proposed. The idea is to view the Factorization Paradox through the lens of effective field theory (EFT), and group all non-factorizating configurations in our complete theory into two collections: those which can be described in the EFT as smooth, geometric wormholes, and those which cannot. We will refer to the second class as \textit{stringy wormholes}, which could be Planckian, non-geometric, or something even wilder. While the net sum of all geometric and stringy wormholes vanishes by assumption, the sum over geometric wormholes on its own need not vanish. Thus, in order to define an effective gravitational path integral that matches the microscopics, we must include additional non-localities beyond the geometric wormholes in order to parametrize the effects of stringy wormholes that have been excluded from our effective description.

In this paper, we propose a tight formal analogy between these additional non-localities, which we dub \textit{counter-wormholes}, and the counterterms needed to cancel ultraviolet (UV) divergences in the perturbative calculation of QFT observables.\footnote{Our original motivation for this analogy was the observation that loop divergences arise from shrinking a loop to zero size, and thus inherently involve topology change on the particle worldline. See also e.g. \cite{Anous:2020lka,Post:2022dfi} for recent discussions of a possible analogy between topology change and perturbative quantum field theory.} If we attempted to use finite, effective coupling constants to directly evaluate Feynman graphs, we would obtain physically unreasonable answers involving UV divergences from the integration over high-energy modes of our EFT fields. Analogously, if we attempted to define a gravitational path integral including only those wormholes visible in the EFT, we would obtain physically unreasonable answers involving non-factorization or ensembles arising from wormhole contruibutions in our effective gravitational path integral.

In both cases, these physically unreasonable answers should not be taken too seriously. They are not an inconsistency of the theory, but rather signal the importance of some UV effects we have incorrectly ignored in our effective description.\footnote{See \cite{Hernandez-Cuenca:2024pey} for further comments on the relationship between spacetime wormholes and the coarse-graining of UV-complete theories.} In the case of perturbation theory, the resolution is to modify our naive action by the addition of UV divergent counterterms which precisely cancel against the UV divergences in Feynman graphs. If we then calculate using a renormalized perturbation theory that includes the counterterms, we obtain finite answers for physically observable quantities which can be matched to our effective description. Analogously, we view the prescriptions of \cite{Saad:2021rcu,Saad:2021uzi,Blommaert:2021fob} as instances of a \textit{renormalized gravitational path integral}, where we have modified the naive sum over smooth spacetimes by the addition of counter-wormholes in order to obtain physically reasonable, factorizing answers for disconnected correlation functions.

This approach to the Factorization Paradox has mainly been studied on a case-by-case basis. Our goal in proposing an analogy with perturbative renormalization is to look for precise gravitational analogs of well-known structures that control the systematics of renormalization. In particular, the perturbative calculation of counterterms is controlled by an algebraic structure, the Connes-Kreimer Hopf algebra $\mathcal{H}_{\rm CK}$ \cite{Connes:1998qv,Connes:1999yr,Connes:2000fe}, which formalizes the recursive Bogoliubov-Parasiuk-Hepp-Zimmermann (BPHZ) procedure \cite{Bogoliubow_Parasiuk_1957,Hepp:1966eg,Zimmermann:1969jj}.\footnote{See Section \ref{sec:PERT_RENORM} or \cite{CM1} for a review.} This Hopf algebra efficiently encodes the combinatorial structure of divergences in Feynman graphs, taking into account the possibility of sub-divergences: sub-graphs of a given graph which are independently divergent, and which are already renormalized by a local counterterm at an earlier stage of the recursive procedure.

In an extremely simple gravitational path integral \cite{MaMax}, we explain how the calculation of counter-wormholes is controlled by an analogous Hopf algebra, which turns out to be the well-known Faà di Bruno Hopf algebra $\mathcal{H}_{\rm FdB}$. We show how $\mathcal{H}_{\rm FdB}$ efficiently encodes the combinatorics of multi-boundary wormholes and systematizes the construction of a factorizing renormalized gravitational path integral. The analogs of sub-divergences are sub-wormholes: embedded submanifolds of a given spacetime which are themselves wormholes, and which might already be canceled by counter-wormholes at an earlier stage of the gravitational analog of the BPHZ procedure.\footnote{A very similar recursive algorithm was used in \cite{Blommaert:2021fob} in order to achieve all-order factorization.} While the toy model we consider is quite simple, our hope is to identify algebraic structures that might serve as useful organizing structures for understanding factorization in more realistic theories.

One immediate upshot of this result is the existence of a symmetry group of the renormalized gravitational path integral. In the case of perturbative renormalization, the group $G_{\rm CK} = {\rm Spec}(\mathcal{H}_{\rm CK})$ dual to the Connes-Kreimer Hopf algebra acts as a group of symmetries of renormalized perturbation theory, reorganizing the sum over Feynman diagrams and counterterms.\footnote{Recall that the spectrum ${\rm Spec}(A)$ of a commutative $\mathbb{C}$-algebra $A$ is the set of algebra homomorphisms $A \to \mathbb{C}$. See \ref{prop:dual_group_to_Hopf_algebra} for the definition of the group structure on ${\rm Spec}(\mathcal{H})$ for a commutative Hopf algebra $\mathcal{H}$.} Analogously, the group $G_{\rm FdB} = {\rm Spec}(\mathcal{H}_{\rm FdB})$ acts as a symmetry of the gravitational path integral, reorganizing the sum over spacetime manifolds and counter-wormholes. One instance of this group action is the integration out of microscopic wormholes into non-local effects as desribed by Coleman, Giddings, and Strominger \cite{Col,GS1,GS2}. Another is the cancelation of microscopic, stringy wormholes against smooth, geometric wormholes, which can be interpreted as a form of ER = EPR \cite{Maldacena:2013xja}.\footnote{See \cite{JMkitp} for further comments on the meaning of ER = EPR in spacetime, as opposed to merely in space, and its role in gauge fixing the sum over topologies.}

More generally, we view the action of $G_{\rm FdB}$ as realizing the gravitational gauge redundancies of \cite{Jafferis:2017tiu, MaMax} in our toy model, illustrating explicitly that the different gauge fixings described in \cite{MaMax} are related by the action of $G_{\rm FdB}$. These gravitational gauge redundancies between spacetimes of distinct topology go hand in hand with the cancelation of non-factorizing contributions,\footnote{In general, gauge redundancies and detailed cancelations are closely related: a gauge redundancy implies the cancelation of anything non-gauge-invariant, and conversely, a cancelation suggests the existence of a gauge fixed description in which the cancelations are made manifest.} and seem to be essential to the resolution of the Factorization Paradox.\footnote{These gauge redundancies also plays a key role in understanding how bulk EFT could remain valid in a restricted sense behind old black hole horizons \cite{Akers:2022qdl}.} In most cases, these gravitational gauge redundancies have been described rather abstractly via the existence of null states. We find it intriguing that in a toy model these gauge redundancies are realized concretely by a group action, and hope this structure persists in general.

The rest of this paper is organized as follows. In Section \ref{sec:PERT_RENORM}, we review the BPHZ approach to perturbative renormalization and its formalization in terms of the Connes-Kreimer Hopf algebra, highlighting the physical and mathematical structures relevant for our analogy. In Section \ref{sec:GRAV_BPHZ}, we describe the very simple toy model of \cite{MaMax} and explicitly carry out the gravitational analog of the BPHZ procedure. We also show that the algebraic structure of this gravitational BPHZ is captured by the Faà di Bruno Hopf algebra $\mathcal{H}_{\rm FdB}$ and its dual group $G_{\rm FdB}$. Finally in Section \ref{sec:DISC}, we comment on the lessons that may be drawn from our results, and discuss some possible extensions and applications.

The analogy drawn in this paper between the gravitational path integral and perturbative renormalization is summarized in Table \ref{tab:analogy}.

\begin{table}[!ht]
    \centering
    \begingroup
    \renewcommand{\arraystretch}{2.2}
    \scalebox{.8}{\begin{tabular}{cc}
    \toprule
\textbf{Gravity} & \textbf{Renormalization} \\
\cmidrule(lr){1-2}
Multiboundary partition functions &  1PI effective action\\
\cmidrule(lr){1-2}
Semiclassical wormholes &  Divergent Feynman diagrams\\
\cmidrule(lr){1-2}
Sub-wormholes & Sub-divergences\\
\cmidrule(lr){1-2}
Stringy wormholes & Counterterms  \\
\cmidrule(lr){1-2}
Faà di Bruno Hopf algebra & Connes--Kreimer Hopf algebra\\
\bottomrule
\end{tabular}}
\endgroup
\caption{The analogy between the gravitational path integral and perturbative renormalization.}
    \label{tab:analogy}
\end{table}

\section{Review of perturbative renormalization}\label{sec:PERT_RENORM}

In this section, we briefly review the BPHZ approach to perturbative renormalization \cite{Bogoliubow_Parasiuk_1957,Hepp:1966eg,Zimmermann:1969jj} and its algebraic formalization \cite{Connes:1998qv,CM1} in terms of the Connes-Kreimer Hopf algebra, with an eye towards our analogy with the gravitational path integral.\footnote{We assume the reader is somewhat familiar with the basics of perturbative renormalization. For a comprehensive review, see \cite{CM1}.} First, in Section \ref{sec:DIV_AND_COUNTER}, we recall the basic idea of renormalization. Next, in Section \ref{sec:SUBDIV_BPHZ}, we review the BPHZ procedure for handling sub-divergences. Finally, in Section \ref{sec:CONNES_KREIMER}, we review the formalization of perturbative renormalization in terms of the Connes-Kreimer Hopf algebra.

\subsection{Divergences and counterterms}\label{sec:DIV_AND_COUNTER}

The starting point of perturbative renormalization is a $D$-dimensional local action functional of some dynamical fields $\phi$:
\begin{equation}
    S(\phi) = \int \mathrm{d}^D x\ \mathcal{L}(\phi).
\end{equation}
For notational simplicity we consider a single scalar field $\phi$. We take the Lagrangian $\mathcal{L}(\phi)$ to be of the form
\begin{equation}
    \mathcal{L}(\phi) = \frac{1}{2} \left( \partial \phi \right)^2 - \frac{1}{2} m^2 \phi^2 + \mathcal{L}_{\rm int}(\phi).
\end{equation}
The interaction Lagrangian $\mathcal{L}_{\rm int}(\phi)$ is assumed to be a polynomial (or power series) in $\phi$ and its derivatives. For example, we might consider $\phi^3$ theory, defined by
\begin{equation}\label{eq:phi_3}
    \mathcal{L}_{\rm int}(\phi) = \frac{1}{3!} g \phi^3.
\end{equation}
Each possible monomial in the interaction Lagrangian comes with a coupling constant, such $g$ above, which we view as tunable parameters defining the theory.\footnote{For field strength and mass renormalization, we must include the possibility of quadratic terms in $\mathcal{L}_{\rm int}(\phi)$.}

The action functional $S(\phi)$ defines a (Euclidean) QFT through the formal path integral
\begin{equation}\label{eq:path_integral}
    \mathcal{Z} = \int \mathcal{D}\phi\ e^{- S(\phi)/\hbar}.
\end{equation}
Physical observables are computed perturbatively in $\hbar$ via a sum over Feynman graphs $\Gamma$, whose edges are labeled by particle species and whose vertices are labeled by monomials in the interaction Lagrangian $\mathcal{L}_{\rm int}$. Each Feynman graph $\Gamma$ defines an integral
\begin{equation}
    U(\Gamma)(p_1, \dots, p_N) = \int \frac{{\rm d}^D k_1}{(2 \pi)^D} \cdots \frac{{\rm d}^D k_L}{(2 \pi)^D}\ I_\Gamma(p_1, \dots, p_N, k_1, \dots, k_L),
\end{equation}
over unconstrained loop momenta $k_i$, refered to as the \textit{unrenormalized Feynman integral} of the graph $\Gamma$. Here $N$ is the number of external legs of $\Gamma$ and $L$ is the loop number (first Betti number) of $\Gamma$. The integrand $I_\Gamma$ is computed via truncated \textit{Feynman rules}, which assign a propagator to each internal edge and an interaction term to each vertex. Our convention is that the powers of $\hbar$ and the couplings $g_i$ are excluded from the integrals $U(\Gamma)$, for reasons that will be important later; to recover the more standard values assigned to a Feynman graph $\Gamma$, we must multiply $U(\Gamma)$ by appropriate powers of the $\hbar$ and the couplings, as in \eqref{eq:S_eff} below. For example, in $\phi^3$ theory \eqref{eq:phi_3}, the Feynman graph $\Gamma$ illustrated in Figure \ref{fig:diagram} represents the integral
\begin{equation}
    U(\Gamma)(p) = \int \frac{{\rm d}^D k}{(2 \pi)^D}\ \left(\frac{1}{k^2 + m^2}\right) \left(\frac{1}{(p+k)^2 + m^2}\right).
    \label{eq:diagram}
\end{equation}
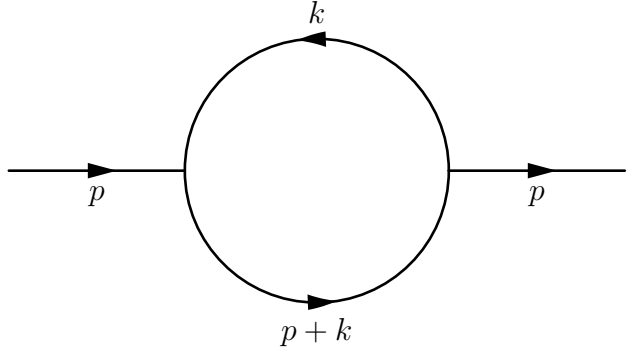
\begin{figure}
\centering
\begin{fmffile}{45}
\begin{fmfgraph*}(233,144)
     \fmfleft{i}
     \fmfright{o}
     \fmf{fermion,tension=3,label=$p$}{i,v1}
     \fmf{fermion,left=1,right,label=$p+k$,label.side=top}{v1,v2}
     \fmf{fermion,left=1,right,label=$k$,label.side=bottom}{v2,v1}
     \fmf{fermion,tension=3,label=$p$}{v2,o}
\end{fmfgraph*}
\end{fmffile}
\caption{A simple Feynman diagram $\Gamma$ in $\phi^3$ theory presented in \cite{CM1}, whose unrenormalized Feynman integral $U(\Gamma)$ is given by \eqref{eq:diagram}. The integration is over the internal momentum $k$ that runs through the loop.}
\label{fig:diagram}
\end{figure}

An efficient way to encode all physical observables of our QFT at once is through the \textit{effective action} $S_{\rm eff}(\phi)$, computed via a sum over all \textit{one-particle irriducible} (1PI) Feynman graphs as
\begin{equation}\label{eq:S_eff}
    S_{\rm eff}(\phi) = S(\phi) - \sum_{\Gamma \in {\rm 1PI}} \hbar^L (- g_I)^{n_I} \frac{U(\Gamma)(\phi)}{\sigma(\Gamma)}.
\end{equation}
There are various things to explain in this equation. First of all, for a given graph $\Gamma$, $L$ is the loop number and $n_I = (n_1, n_2, \dots)$ are the numbers of vertices associated to couplings $g_i$. We have further abbreviated $(-g_I)^{n_I} = \prod_i (-g_i)^{n_i}$ for the product of all couplings associated to vertices in $\Gamma$. The minus signs in \eqref{eq:S_eff} arise from the minus sign in the path integral \eqref{eq:path_integral}. The symmetry factor $\sigma(\Gamma)$ is the order of ${\rm Aut}(\Gamma)$ defined holding external legs fixed. Finally, the functional $U(\Gamma)(\phi)$ is defined in terms of the momentum-space fields $\phi(p)$ by
\begin{equation}
    U(\Gamma)(\phi) = \frac{1}{N!} \int_{\sum p_i = 0} \frac{{\rm d}^D p_1}{(2 \pi)^D} \cdots \frac{{\rm d}^D p_N}{(2 \pi)^D}\ \phi(p_1) \cdots \phi(p_N)\ U(\Gamma)(p_1, \dots, p_N).
\end{equation}
An important note is that the relevant notion of 1PI Feynman graph for the sum in \eqref{eq:S_eff} excludes trees, since the first term $S(\phi)$ can be interpreted as the contribution of tree graphs to the effective action $S_{\rm eff}(\phi)$. The benefit of the effective action is that it reproduces, at tree-level, the full perturbative expansion of physical observables computed from $S(\phi)$ at all loops.

Famously, if we naively take our coupling constants in the action $S(\phi)$ to be finite numbers, the integrals $U(\Gamma)$ will contain UV divergences from the integration over large values of the loop momenta $k_i$. These UV divergences will show up in the calculation of the effective action $S_{\rm eff}(\phi)$, producing unreasonable answers for physical observables. The basic idea of renormalization is that the bare coupling constants appearing in the path integral are not directly observable. Thus, demanding that they be finite is a mis-application of the relevant physical condition: the couplings that need to be finite are the effective couplings in $S_{\rm eff}(\phi)$, not the bare couplings. We are free to take the bare couplings to be divergent, if doing so allows us to produce finite answers to physically meaningful questions.

Suppose, then, that we are given an action functional $S(\phi)$ with finite couplings, and we seek to use it to produce a physically sensible theory. Rather than directly using $S(\phi)$ to define a path integral as in \eqref{eq:path_integral}, we must first modify our action by the addition of UV divergent \textit{counterterms} to define a \textit{bare action} $S_{\rm bare}(\phi)$.\footnote{A possible point of contention is whether the counterterms are something to be added or whether they were always there in the first place. As discussed in the introduction, the physically meaningful answer is that they must have always been there, as the counterterms represent unknown UV physics needed to render the theory sensible. However, our starting point here is more formal, given by the problem of defining a sensible physical theory out of an action principle with finite couplings. This will be analogous to the formal problem of defining a physically sensible theory out of a given gravitational path integral.} These counterterms must be precisely tuned so that the UV divergences they produce exactly cancel against every UV divergence arising from loops. Once appropriate counterterms are chosen, we may then define a path integral using the bare action, and derive a finite effective action as in \eqref{eq:S_eff} by using the Feynman rules of $S_{\rm bare}(\phi)$.

Practically speaking, to compute appropriate counterterms, one must first choose a regularization scheme, say dimensional regularization (dimreg), in order to regulate the divergences in loop integrals. This yields finite values for the unrenormalized Feynman integrals $U(\Gamma)$ at the cost of dependence on an unphysical cutoff parameter $\varepsilon$, which must be taken to zero at the end of the calculation. The counterterms are also taken to be functions of $\varepsilon$, and the goal is to choose counterterms such that the physical observables computed with counterterms included yield finite values in the limit $\varepsilon \to 0$. To make this choice of counterterms, we must first choose a subtraction scheme, such as minimal subtraction (MS), that specifies the ``divergent part'' of any regulated loop integral.

\subsection{Sub-divergences and the BPHZ procedure}\label{sec:SUBDIV_BPHZ}

The essential insight of the BPHZ procedure is that counterterms, like physical observables, can be computed perturbatively via a sum over Feynman graphs. To each 1PI Feynman graph $\Gamma$, we associate a counterterm $C(\Gamma)$ which is chosen in order to precisely cancel UV divergences in $U(\Gamma)$.\footnote{As with the unrenormalized Feynman integrals $U(\Gamma)$, our convention is that $C(\Gamma)$ excludes powers of $\hbar$ and the couplings, which must be reintroduced to obtain the more conventional values of the counterterms.} Each counterterm defines a functional $C(\Gamma)(\phi)$, which is normally required to be local. The bare action is then defined via the addition of the counterterms for every 1PI graph $\Gamma$, as follows.
\begin{equation}\label{eq:S_bare}
    S_{\rm bare}(\phi) = S(\phi) - \sum_{\Gamma \in {\rm 1PI}} \hbar^L (- g_I)^{n_I} \frac{C(\Gamma)(\phi)}{\sigma(\Gamma)}.
\end{equation}
Note the similarity to the calculation of the effection action \eqref{eq:S_eff}.

An essential complication is the possibility of sub-divergences: sub-graphs $\gamma \subset \Gamma$ of a given Feynman graph which are divergent on their own. Figure \ref{fig:diagram0} presents a simple example of Feynman graph with a sub-divergence.
\begin{figure}
\centering
\includegraphics[scale=0.5]{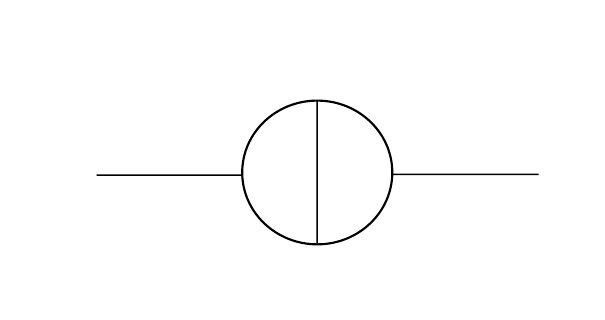}
\caption{A simple example of diagram with sub-divergences in $\phi^3$ theory. This diagram has loops nested inside other loops, which is the source of the recursive nature of the BPHZ procedure.}
\label{fig:diagram0}
\end{figure}
While the divergent subgraph $\gamma$ will certainly contribute to the divergences in $U(\Gamma)$,\footnote{In fact, it will contribute in a severe way: the divergent piece will be a non-local functional of the external momenta in $\Gamma$ which could not be canceled by any local counterterm $C(\Gamma)(\phi)$.} these divergences will already be taken into account by the counterterm $C(\gamma)$ we have chosen to cancel the divergence in $U(\gamma)$. If we were to include this divergence when computing $C(\Gamma)$, we would overcount the divergences, and fail in defining finite physical observables. Thus, even in dimreg with minimal subtraction, the counterterm $C(\Gamma)$ will not merely be given by the divergent piece of $U(\Gamma)$ when $\Gamma$ contains a divergent subgraph.

Instead, to compute the counterterm for any given Feynman graph $\Gamma$, we must follow the \textit{BPHZ procedure} \cite{Bogoliubow_Parasiuk_1957,Hepp:1966eg,Zimmermann:1969jj}, a recursive algorithm for computing counterterms while taking sub-divergences into account. To compute $C(\Gamma)$, we first identify all divergent subgraphs $\gamma \subset \Gamma$,\footnote{The subgraph $\gamma$ is assumed to be a non-empty, proper subgraph. However, $\gamma$ can be disconnected. If this happens, we define $U(\gamma)$ and $C(\gamma)$ to be the product of their values on each connected component of $\gamma$.} and recursively compute their counterterms $C(\gamma)$ via the BPHZ procedure. Then, we compute a \textit{prepared Feynman integral} $\overline{R}(\Gamma)$ by summing all relevant contracted graphs $\Gamma/\gamma$ weighted by the previously computed counterterms.\footnote{\label{footnote:external_data}In \eqref{eq:prepared_value}, each connected Feynman graph with $N$ external legs should be equipped with the data of a specific monomial of degree $N$ in the fields $\phi$, and the sum over subgraphs includes a sum over this data. The vertices in the contracted graph $\Gamma/\gamma$ obtained by contracting a component of $\gamma$ are labeled by the monomial chosen for $\gamma$.}
\begin{equation}\label{eq:prepared_value}
    \overline{R}(\Gamma) = U(\Gamma) + \sum_{\gamma \subset \Gamma} C(\gamma) U(\Gamma/\gamma).
\end{equation}
The prepared Feynman integral $\overline{R}(\Gamma)$ may still be divergent, but will only include those divergences in $\Gamma$ that do not arise from any subgraph. We may then compute the counterterm $C(\Gamma)$ to be the divergent piece of $\overline{R}(\Gamma)$, as determined by our subtraction scheme, so that the \textit{renormalized Feynman integral}
\begin{equation}\label{eq:renorm_value}
    R(\Gamma) = \overline{R}(\Gamma) + C(\Gamma)= C(\Gamma) + U(\Gamma) + \sum_{\gamma \subset \Gamma} C(\gamma) U(\Gamma/\gamma),
\end{equation}
is finite.

Once we have computed counterterms via the BPHZ procedure, we now have three distinct yet equivalent ways to incorporate them into the calculation of physical observables.

\textbf{First approach:} The first, most direct approach is to replace the unrenormalized Feynman integrals $U(\Gamma)$ with the renormalized ones $R(\Gamma)$ when calculating physical observables. In particular, we can replace $U(\Gamma)$ with $R(\Gamma)$ in \eqref{eq:S_eff}, in order to compute a physically reasonable effective action directly from $S(\phi)$:
\begin{equation}\label{eq:S_eff_renorm}
    S_{\rm eff}(\phi) = S(\phi) - \sum_{\Gamma \in {\rm 1PI}} \hbar^L (- g_I)^{n_I} \frac{R(\Gamma)(\phi)}{\sigma(\Gamma)}.
\end{equation}
This approach has the benefit of making the cancelation of UV divergences manifest, since each renormalized value $R(\Gamma)$ is independently finite. However, it has the downside of seeming a bit arbitrary. We have modified the way we evaluate Feynman diagrams, only including some finite piece of the Feynman integrals in our calculation.

\textbf{Second approach:} The second approach is to add the counterterms $C(\Gamma)$ to the action, in order to define a bare action as in \eqref{eq:S_bare}. We then forget where the bare action came from, and recompute physical observables using the ``unrenormalized Feynman integrals'' $U(\Gamma)$ computed with the Feynman rules derived from $S_{\rm bare}$. Graph by graph, we will encounter UV divergences which will cancel out of any physical observable, seemingly miraculously. While this approach obscures the cancelation of UV divergences, it makes manifest that we are playing the same sort of game as before: our theory is defined via the path integral by some local action functional, which simply happens to have UV divergent couplings. Notably in this approach, the cancelation occurs between Feynman diagrams of different loop orders.

\textbf{Third approach:} The third approach is somewhere between the first two, and consists of enlarging the set of Feynman graphs we include in our perturbative expansion. In addition to vertices labeled by interaction terms in $S(\phi)$, we now have additional vertices labeled by the counterterms. Now, each Feynman diagram will still be individually divergent, but the combinatorics of formula \eqref{eq:renorm_value} will be more readily visible, since we can keep track of the contributions of each counterterm.

Crucially, while these three approaches organize the perturbative expansion differently, they produce exactly equal answers for physical observables. The only difference is whether to group the counterterms $C(\Gamma)$ with the graphs $\Gamma$ from which they arise, add them all to the action, or leave them as an explicit, additional type of contribution to the perturbative expansion. This redundancy in how we organize the perturbative expansion is a form of gauge redundancy (or duality): we have multiple, exactly equivalent descriptions of the same physics. To see this redundancy more explicitly, we turn now to the algebraic formalization of renormalized perturbation theory.

\subsection{The Connes--Kreimer Hopf algebra}\label{sec:CONNES_KREIMER}

In the algebraic formalization of renormalization, the BPHZ procedure described above is efficiently packaged in terms of the Connes-Kreimer Hopf algebra $\mathcal{H}_{\rm CK}$, a commutative Hopf algebra over $\mathbb{C}$. As reviewed in Appendix \ref{sec:Hopf_algs}, a commutative Hopf algebra $\mathcal{H}$ defines a dual group $G = {\rm Spec}(\mathcal{H})$, whose group law $\star : G \times G \to G$ is defined in terms of the coproduct on $\mathcal{H}$. The group $G_{\rm CK} = {\rm Spec}(\mathcal{H}_{\rm CK})$ is precisely the group underlying the multiple descriptions of renormalized perturbation theory described in the previous section, and allows one to effortlessly pass between them.

In order to motivate the definition of $\mathcal{H}_{\rm CK}$, let us examine a formula that appeared at multiple points in the discussion above. In \eqref{eq:S_eff}, \eqref{eq:S_bare}, and \eqref{eq:S_eff_renorm}, we have seen transformations of the form
\begin{equation}\label{eq:CK_action}
    S(\phi) \to S(\phi) - \sum_{\Gamma \in {\rm 1PI}} \hbar^L (- g_I^S)^{n_I} \frac{F(\Gamma)(\phi)}{\sigma(\Gamma)},
\end{equation}
where $F(\Gamma)$ is some functional on the set of 1PI graphs and the couplings $g_I^S$ are those in the action $S$. The goal of this section is to interpret \eqref{eq:CK_action} as the action of $F \in G_{\rm CK}$ on the space of action functionals. In terms of this group action, the results of the previous section can be summarized via a commutative diagram:
\begin{equation}\label{eq:two_approaches}
   \begin{tikzcd}[column sep=huge]
    S \arrow[r, "C"] \arrow[dr, "R"'] & S_{\rm bare} \arrow[d, "U"] \\
    & S_{\rm eff}
    \end{tikzcd}
\end{equation}
In other words, as described above, we can obtain $S_{\rm eff}$ in three ways: by acting on $S$ with $R$, by acting on $S_{\rm bare}$ by $U$, or by acting on $S$ first with $C$ and then with $U$.

Our goal, then, is to define $\mathcal{H}_{\rm CK}$ such that $R, C,$ and $U$ are elements of the dual group $G_{\rm CK}$ of algebra homomorphisms out of $\mathcal{H}_{\rm CK}$, and further such that we have
\begin{equation}
    R = C \star U,
\end{equation}
under the group law of $G_{\rm CK}$.\footnote{The order of multiplication is chosen so that the action of $G_{\rm CK}$ on the space of action functionals \eqref{eq:CK_action} is a \textit{right} action.} We have already seen how $R, C$ and $U$ are related in \eqref{eq:renorm_value}: in order to evaluate $R$ on a Feynman graph $\Gamma$, we extract all possible subgraphs $\gamma \subset \Gamma$, and evaluate $C$ and $U$ on the subgraphs and quotient graphs respectively. This motivates the following definition.

\begin{defn}
    For a given QFT, the \textit{Connes--Kreimer Hopf algebra} $\mathcal{H}_{\rm CK}$ is the free commutative algebra over $\mathbb{C}$ generated by the set of 1PI Feynman graphs $\Gamma$ (excluding trees),\footnote{As in Footnote \ref{footnote:external_data}, each Feynman graph $\Gamma$ with $N$ external legs must be equipped with a choice of a monomial of degree $N$ in the fields $\phi$. This data is summed over in \eqref{eq:CK_coprod}. We will not keep careful track of this data in our exposition; for details, see \cite{CM1}.} equipped with counit $\varepsilon:\mathcal{H}_{\rm CK} \to \mathbb{C}$ and coproduct $\Delta : \mathcal{H}_{\rm CK} \to \mathcal{H}_{\rm CK} \otimes \mathcal{H}_{\rm CK}$ defined on algebra generators $\Gamma$ by $\varepsilon(\Gamma) = 0$ and
    \begin{equation}\label{eq:CK_coprod}
        \Delta(\Gamma) = \Gamma \otimes 1 + 1 \otimes \Gamma + \sum_{\gamma \subset \Gamma} \gamma \otimes \Gamma/\gamma,
    \end{equation}
    and extended multiplicatively.
\end{defn}

\begin{prop}[\cite{Connes:1999yr}]\label{prop:CK_is_Hopf}
    As defined above, $\mathcal{H}_{\rm CK}$ is a Hopf algebra.
\end{prop}

\begin{proof}
    We follow the proof of \cite[Theorem 1.27]{CM1}. By definition, $\varepsilon$ and $\Delta$ are algebra homomorphisms, and it is immediate from equation \eqref{eq:CK_coprod} that $\varepsilon$ is a two sided counit for $\Delta$. For coassociativity, observe that both $(\Delta \otimes {\rm Id})(\Delta(\Gamma))$ and $({\rm Id} \otimes \Delta)(\Delta(\Gamma))$ can be written as a sum over nested subgraphs,
    \begin{equation}\label{eq:CK_coprod2}
        \Delta^2(\Gamma) = \sum_{\gamma \subset \gamma' \subset \Gamma} \gamma \otimes \gamma'/\gamma \otimes \Gamma/\gamma',
    \end{equation}
    where now the sum includes both empty and improper subgraphs and we interpret the empty graph as the unit $1 \in \mathcal{H}_{\rm CK}$.

    We have shown that $\mathcal{H}_{\rm CK}$ is a bialgebra. The existence of an antipode follows from Proposition \ref{prop:conn_graded} upon noting that $\mathcal{H}_{\rm CK}$ is a connected graded bialgebra, where the grading of a generator $\Gamma$ is its loop number.
\end{proof}

\begin{defn}
    The \textit{Connes--Kreimer group} $G_{\rm CK}$ is the spectrum ${\rm Spec}(\mathcal{H}_{\rm CK})$ (see \cite{Connes:1998qv,CM1}). The multiplication of two elements $F, G \in G_{\rm CK}$ (given by arbitrary maps from the set of 1PI Feynman graphs to $\mathbb{C}$) is defined by \begin{equation}\label{eq:CK_group_law}
    (F \star G)(\Gamma) = F(\Gamma) + G(\Gamma) + \sum_{\gamma \subset \Gamma} F(\gamma) G(\Gamma/\gamma),
    \end{equation}
    as in \eqref{eq:renorm_value} for the product $R = C \star U$.
\end{defn}

\begin{prop}[\cite{Connes:2000fe}]\label{thm:CK_action} The formula \eqref{eq:CK_action} defines a right action of $G_{\rm CK}$ on the space of action functions $S(\phi)$.
\end{prop}

\begin{proof}
    Clearly, the identity element $1 \in G_{\rm CK}$ given by the counit $\varepsilon: \mathcal{H}_{\rm CK} \to \mathbb{C}$ acts trivially, since $\varepsilon(\Gamma) = 0$ for all 1PI graphs $\Gamma$. Now, for $F, G \in G_{\rm CK}$, we have that
    \begin{equation}
        (S \cdot (F \star G))(\phi) = S(\phi) - \sum_{\Gamma \in {\rm 1PI}} \hbar^L (- g_I^S)^{n_I} \frac{(F\star G)(\Gamma)(\phi)}{\sigma(\Gamma)}.
    \end{equation}
    Expanding $F \star G$, we obtain
    \begin{equation}
        (S \cdot (F \star G))(\phi) = S(\phi) - \sum_{\Gamma \in {\rm 1PI}}\sum_{\gamma\subset\Gamma} \hbar^L (- g_I^S)^{n_I} \frac{F(\gamma)(\phi)G(\Gamma/\gamma)(\phi)}{\sigma(\Gamma)},
    \end{equation}
    where to alleviate notations compared to Equation \eqref{eq:CK_group_law}, we have included the cases $\gamma=\emptyset$ and $\gamma=\Gamma$ in the sum.
    This sum can be rewritten as 
    \begin{equation}\label{eq:end_of_proof}
        (S \cdot (F \star G))(\phi) = S(\phi) - \sum_{\Gamma^\prime \in {\rm 1PI}}\hbar^L (- g_I^{S \cdot F})^{n_I} \frac{G(\Gamma^\prime)(\phi)}{\sigma(\Gamma^\prime)},
    \end{equation}
    where \begin{equation}g_I^{S \cdot F}=g_I-\sum_{\gamma \in {\rm 1PI}} \hbar^L (- g_I^S)^{n_I} \frac{F(\gamma)(\phi)}{\sigma(\gamma)},\end{equation}
    are the coupling constants appearing in $S \cdot F$. But \eqref{eq:end_of_proof} is simply the expression for $((S \cdot F) \cdot G)(\phi)$.
\end{proof}

Note that in the proof of Proposition \ref{thm:CK_action}, it was important that the action of $G$ on $S \cdot F$ involved the shifted coupling constants $g_I^{S \cdot F}$ appearing in the action functional $S \cdot F$. This is the reason why we excluded the coupling constants from our definitions of the functionals $U, C, R \in G_{\rm CK}$ above, as elements of $G_{\rm CK}$ must be defined independently from the coupling constants in the action functionals on which they act.

\section{A gravitational BPHZ procedure}\label{sec:GRAV_BPHZ}

In the previous section, we summarized the BPHZ procedure underlying perturbative renormalization, and described the underlying Hopf algebraic structure discovered by Connes and Kreimer. We now turn to formulating an analogous procedure to systematically restore factorization in simple two-dimensional topological theories of gravity by the addition of \textit{counter-wormholes}, explicit non-localities in the bulk action.

We will first review a topological theory based on a gravitational path integral due to Marolf and Maxfield \cite{MaMax} which fails to factorize due to the presence of spacetime wormholes. We will then show that this theory can be systematically corrected into a factorizing theory by introducing counter-wormholes, following a precise analog of the BPHZ procedure underlying perturbative renormalization. This procedure is quite similar to that considered in \cite{Blommaert:2021fob}; our goal here is to illuminate the algebraic structures controlling this approach to resolving the Factorization Paradox.

\subsection{Marolf--Maxfield theory}

In \cite{MaMax}, Marolf and Maxfield introduced an exactly solvable topological theory of gravity in two dimensions. The observables of this theory are the $n$-boundary correlation functions $\braket{Z^n}$, defined via a discrete path integral over the set $\mathcal{V}_n$ of (diffeomorphism classes of) oriented smooth surfaces $\Sigma$ with $n$ labeled circular boundaries:
\begin{equation}
   \braket{Z^n} = \mathfrak{Z}^{-1} \sum_{\Sigma \in \mathcal{V}_n} \mu(\Sigma)e^{S_0\tilde{\chi}(\Sigma)}.
    \label{eq:MaMaxDef}
\end{equation}
Here, $\mathfrak{Z}$ is a normalizing factor, $\tilde{\chi}(\Sigma) = \chi(\Sigma) + n$ is a modified Euler characteristic of $\Sigma$, and $\mu(\Sigma)$ is a symmetry factor, defined by
\begin{equation}
\mu(\Sigma)=\frac{1}{\prod_gm_g!},
\end{equation}
where the $m_g$ are the numbers of \textit{closed universes} (connected components of $\Sigma$ with no boundary) of genus $g$. The modified Euler characteristic $\tilde{\chi}(\Sigma)$ is independent of the number of boundaries of $\Sigma$, since adding a boundary decreases $\chi(\Sigma)$ by one while increasing $n$ by one.

Following \cite{MaMax}, we may simplify \eqref{eq:MaMaxDef} by integrating out all higher genus surfaces and closed universes. For a given connected component of $\Sigma$, summing over genus gives a contribution $\lambda$ to the path integral, where
\begin{equation}\label{eq:lambda}
    \lambda = \sum_{g = 0}^\infty e^{- S_0 (2 - 2 g)} = \frac{e^{2 S_0}}{1 - e^{-2 S_0}}.
\end{equation}
Moreover, the sum over closed universes exponentiates into an overall prefactor $e^\lambda$, which can be absorbed by choosing $\mathfrak{Z} = e^\lambda$. All that remains is to sum over the possible genus zero surfaces $\Sigma$ connected to the boundary. Let $\mathcal{W}_n$ denote the subset of $\mathcal{V}_n$ consisting of genus zero surfaces with $n$ boundaries with no closed components. We have
\begin{equation}\label{eq:MaMaxAct}
    \braket{Z^n} = \sum_{\Sigma \in \mathcal{W}_n} \lambda^{k_\Sigma},
\end{equation}
where $k_\Sigma$ is the number of connected components of $\Sigma$. Note that an element of $\mathcal{W}_n$ is uniquely specified by a partition of the set of boundary components, so we may equivalently write \eqref{eq:MaMaxAct} as a sum over partitions of $n$.

Evaluating \eqref{eq:MaMaxAct} as in \cite{MaMax} gives
\begin{equation}\label{eq:Poisson_moments}
    \braket{Z^n} = B_n(\lambda),
\end{equation}
where $B_n$ is the Bell polynomial of order $n$. See Figure \ref{fig:z3mm} for an illustration of the computation of $\braket{Z^3}$ in a Marolf--Maxfield theory. 
\begin{figure}
    \centering
    \includegraphics[scale=0.7]{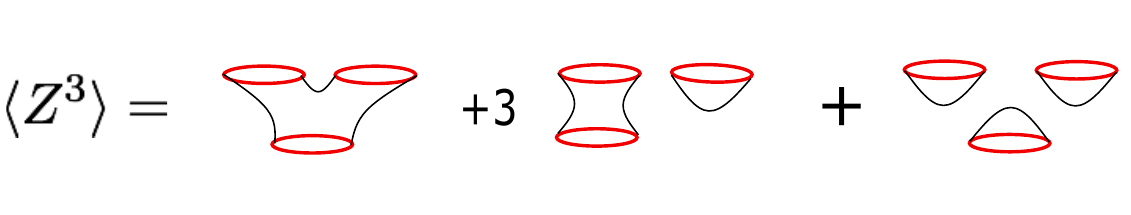}
    \caption{The computation of the spacetime amplitude $\braket{Z^3}$ in a Marolf--Maxfield theory. The first topology has 1 connected component so it contributes $\lambda$, the second topology has 2 connected components so it contributes $\lambda^2$ with multiplicity 3, and the last topology has 3 connected component so it contributes $\lambda^3$. We then recover the $n=3$ case of Equation \eqref{eq:Poisson_moments}.}
    \label{fig:z3mm}
\end{figure}
Notably, Marolf--Maxfield theory is not factorizing, and we do not have $\braket{Z^n} = \braket{Z}^n$. This is unsurprising, as Marolf--Maxfield theory contains spacetime wormholes that contribute nonzero amplitudes to the multi-boundary correlation functions $\braket{Z^n}$, with no stringy wormholes to possibly cancel against. The interpretation given in \cite{MaMax} is that \eqref{eq:MaMaxAct} does not define a single quantum theory of gravity, but instead an ensemble of theories, where $Z$ defines a Poisson random variable of mean $\lambda$, whose moments are given by \eqref{eq:Poisson_moments}.

In this paper, we pursue a complementary interpretation motivated by perturbative renormalization, in which we view \eqref{eq:MaMaxAct} as defining only a low energy gravitational EFT. This EFT is good for answering some questions, but breaks down when pushed to answer more detailed questions like the factorization of multi-boundary correlation functions. We assume that this EFT correctly computes the single boundary partition function $\braket{Z} = \lambda$,\footnote{We do not require $\lambda \in \mathbb{N}$. We will comment on this further in Sections \ref{sec:allorders} and \ref{sec:categ}.} which can be interpreted as the partition function of a holographically dual topological quantum mechanics. However, we will modify the calculation of multi-boundary correlation functions by adding additional \textit{counter-wormholes} which parameterize the unknown stringy wormholes needed to guarantee factorization.

Note that the single boundary partition function $\braket{Z}$ already includes contributions from wormholes in the form of higher-genus surfaces in the sum \eqref{eq:lambda}, which can be understood as a renormalization of the parameter $e^{2 S_0}$ to $\lambda = e^{2 S_0} (1 + \mathcal{O}(e^{-2 S_0}))$. In a more complete treatment, we should also modify the calculation of the single-boundary partition function by counter-wormholes, leading to the simultaneous consideration of the renormalization of $e^{2 S_0}$ and the breakdown of factorization. We will ignore this issue in our treatment to simplify the algebra, taking \eqref{eq:MaMaxAct} to be the definition of Marolf--Maxfield theory and considering only genus zero surfaces and their associated counter-wormholes. See \cite{Blommaert:2021fob} for a complete treatment of wormholes and higher-genus surfaces at the same time.

\subsection{Renormalizing Marolf--Maxfield}\label{sec:RENORM_MM}

Now that the stage is set for Marolf--Maxfield theory, we would like to use the intuition gained in Section \ref{sec:PERT_RENORM} to introduce and motivate an analogy between the Factorization Paradox and perturbative renormalization.

The first step in making our analogy is to identify the gravitational analogs of tree-level and loop-level Feynman diagrams in the Marolf--Maxfield theory \eqref{eq:MaMaxAct}. As stated above, we assume that \eqref{eq:MaMaxAct} correctly computes the single boundary partition function $\braket{Z} = \lambda$, arising from the bulk manifold $\Sigma$ given by a single disk. In any correlation function $\braket{Z^n}$, we always have a term $\lambda^n$ arising from $n$ disks, and we may rewrite \eqref{eq:MaMaxAct} as
\begin{equation}\label{eq:grav_tree_and_loops}
    \braket{Z^n} = \lambda^n + \sum_{\Sigma \in \mathcal{W}_n^{\rm wh}} \lambda^{k_\Sigma},
\end{equation}
where $\mathcal{W}_n^{\rm wh}$ is the subset of $\mathcal{W}_n$ such that at least one component of $\Sigma$ is a wormhole, i.e., has more than one boundary. We view the first term in \eqref{eq:grav_tree_and_loops} as ``tree-level,'' requiring no renormalization, and the rest as ``loop-level,'' requiring renormalization due to their failure to factorize. Thus, the analog of loop-level Feynman diagrams are spacetimes that include wormholes.

In the case of perturbative renormalization, we saw that the physically meaningful observables can be efficiently summarized via effective coupling constants, computed perturbatively from the tree-level couplings via \eqref{eq:S_eff}. Analogously, we view \eqref{eq:grav_tree_and_loops} as defining an effective bulk action functional arising from integrating out wormholes. In particular, the observables of Marolf--Maxfield theory can reproduced from a ``tree-level'' calculation involving only spacetimes without wormholes, provided we assign the value $B_{k_\Sigma}(\lambda)$ to a spacetime consisting of $k_{\Sigma}$ disconnected disks rather than merely $\lambda^{k_\Sigma}$. In the case of perturbative renormalization, the effective coupling constants differ from the bare coupling constants only by terms subleading in $\hbar$, the loop-counting parameter. Here, we see that the role of $\hbar$ is being played by $\lambda^{-1}$, as $B_{k_\Sigma}(\lambda)$ agrees with $\lambda^{k_\Sigma}$ up to terms subleading in $\lambda^{-1}$ arising from spacetime wormholes.

In perturbative renormalization, we know that if we naively take the microscopic coupling constants defining our QFT to be finite numbers, then the physically observable effective couplings constants will suffer UV divergences from loops. Here, we see a direct analog: if we take the bulk action functional $e^{- S(\Sigma)} = \lambda^{k_\Sigma}$ to factorize on disconnected bulk spacetimes, then the physically observable multi-boundary correlation functions $\braket{Z^n} = B_n(\lambda)$, analogous to the effective coupling constants, will fail to factorize due to wormholes.

Motivated by \cite{Blommaert:2021fob}, we claim that factorization can be restored if we modify the bulk action functional $e^{-S(\Sigma)} = \lambda^{k_\Sigma}$ of Marolf--Maxfield theory, allowing it to become a non-local functional of $\Sigma$ which we denote $e^{- S_{\rm bare}(\Sigma)}$. This is the analog of passing to a bare action in perturbative renormalization, where we allow coupling constants to be divergent. Schematically, we expect this ``bare bulk action'' to take the form
\begin{equation}\label{eq:non_local_corrections}
    e^{-S_{\rm bare}(\Sigma)} = \lambda^{k_\Sigma} + \text{Non-local corrections}.
\end{equation}
These non-local corrections should be precisely tuned such that the modified gravitational path integral
\begin{equation}\label{eq:modified_grav_path_integral}
    \braket{Z^n} = e^{- S_{\rm bare}(\text{$n$ disks})} + \sum_{\Sigma \in \mathcal{W}_n^{\rm wh}} e^{-S_{\rm bare}(\Sigma)},
\end{equation}
produces factorizing answers, i.e., $\braket{Z^n} = \lambda^n$, due to cancellations between non-localities in the first, ``tree-level'' term, and non-factorizations arising from the sum over wormhole contributions. As in the unmodified Marolf--Maxfield theory \eqref{eq:MaMaxAct}, we make the simplifying assumption that $e^{- S_{\rm bare}(\Sigma)}$ is independent of the number of boundaries of $\Sigma$ and depends only on the number $k_\Sigma$ of connected components.\footnote{This assumption is merely for convenience: relaxing it would not change anything of conceptual consequence but would require keeping track of much more data.}

As described in the Introduction, we view the non-local corrections in \eqref{eq:non_local_corrections} as arising from having already integrated out microscopic, stringy wormholes in a UV complete theory, whose role is to cancel against the geometric wormholes described by the EFT. With this expectation in mind, we will make the ansatz
\begin{equation}\label{eq:counterwormhole_ansatz}
    e^{-S_{\rm bare}(\text{$n$ disks})} = \lambda^n + \sum_{\Sigma \in \mathcal{W}_n^{\rm wh}} \lambda^{k_\Sigma} C(\Sigma),
\end{equation}
where $C(\Sigma)$ is some functional of the set of spacetime manifolds, defined to be multiplicative on connected components.\footnote{Note that $e^{-S_{\rm bare}(\Sigma)}$ for more general $\Sigma$ is defined to be equal to $e^{-S_{\rm bare}(\text{$k_\Sigma$ disks})}$, by our assumption that $e^{- S_{\rm bare}(\Sigma)}$ depends only on the number of components $k_\Sigma$ of $\Sigma$.} Each term $\lambda^{k_\Sigma} C(\Sigma)$ in \eqref{eq:counterwormhole_ansatz} represents the contribution of a \textit{counter-wormhole} (see Figure \ref{fig:counterwh}) associated to the wormhole $\Sigma$. The factors $C(\Sigma)$ must be chosen to precisely cancel the non-factorization arising from $\Sigma$ in the gravitational path integral.
\begin{figure}
\centering
\includegraphics[scale=0.2]{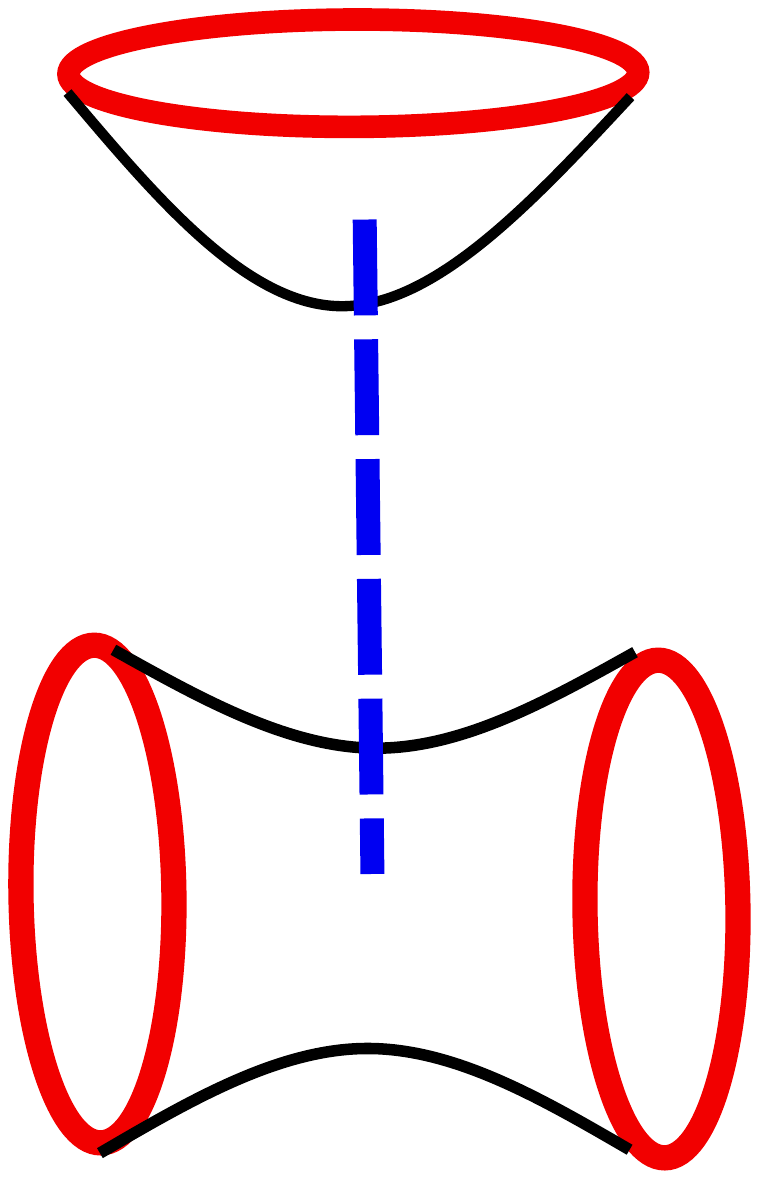}
\caption{A schematic notation for a contribution including a counter-wormhole, represented as a dashed blue line, going between two connected components of spacetime.}
\label{fig:counterwh}
\end{figure}
Crucially, the wormhole contributions in \eqref{eq:modified_grav_path_integral} are computed using the modified action functional $e^{- S_{\rm bare}(\Sigma)}$, rather than the unmodified Marolf--Maxfield action, which themselves include counter-wormhole contributions. Thus, the problem of choosing appropriate counter-wormhole factors $C(\Sigma)$ in order to guarantee factorization is inherently recursive. Now, due to the simplicity of Marolf--Maxfield theory, we could directly solve this recursive problem by hand, in order to obtain $e^{- S_{\rm bare}(\Sigma)}$. However, as our goal is to highlight algebraic structures that may exist in more general gravitational theories, let us instead describe the conceptual root of this recursive subtlety.

In the case of perturbative renormalization, we saw that the BPHZ procedure was needed in order to handle the possibility of divergent subgraphs $\gamma \subset \Gamma$. What are the gravitational analogs of subd-ivergences? Inside a given spacetime manifold $\Sigma$ defining a wormhole contribution to the gravitational path integral, we might have embedded submanifolds $\sigma \subset \Sigma$ which are, themselves, spacetime wormholes (see Figure \ref{fig:subwormhole}).
\begin{figure}
\centering
\includegraphics[scale=0.2]{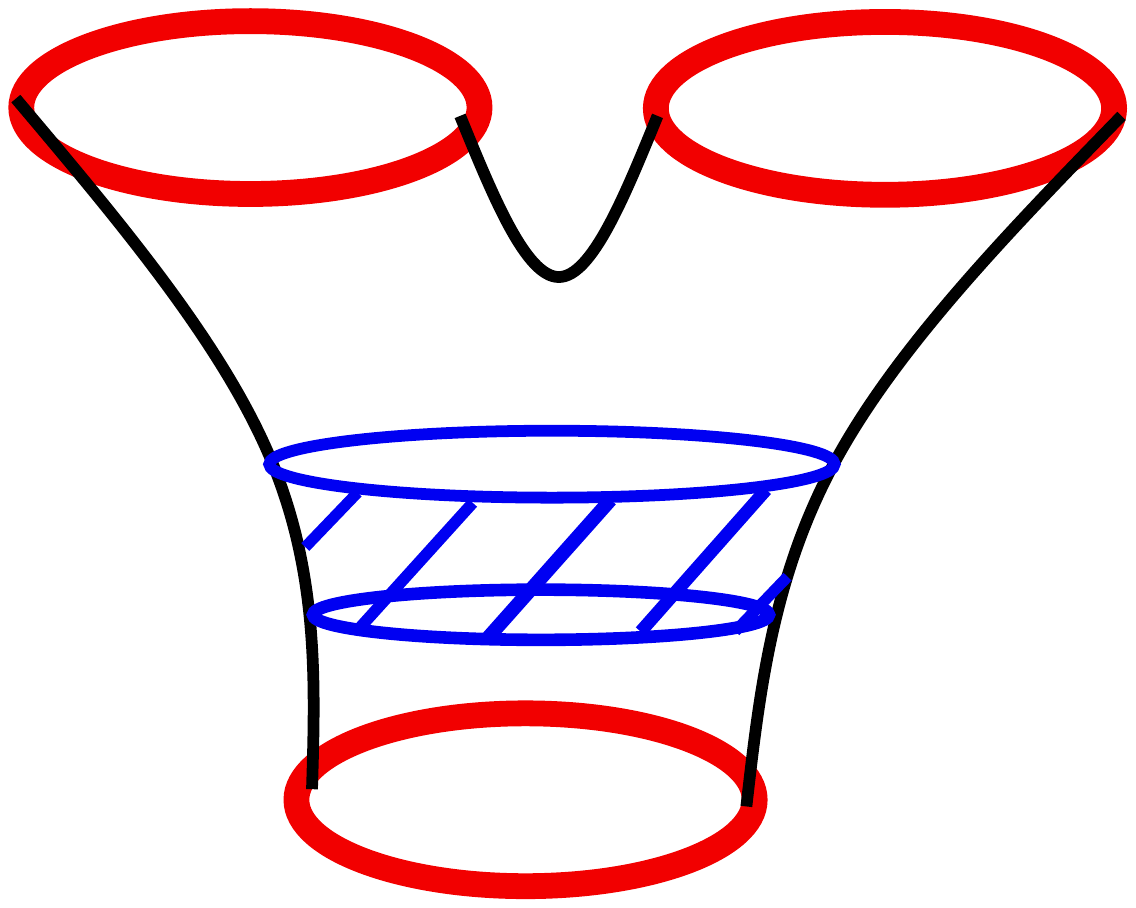}
\caption{An example of sub-wormhole. The pair of pants $\Sigma$ has an embedded cylinder $\sigma\subset \Sigma$ which itself leads to a wormhole contribution.}
\label{fig:subwormhole}
\end{figure}

Presumably, at an earlier stage of the recursive procedure, we would have chosen a counter-wormhole factor $C(\sigma)$ to cancel against the non-factorizations arising from $\sigma$. Thus, not all of the non-factorization in the unmodified contribution of $\Sigma$ to the gravitational path integral is new: some of it must be ascribed to the \textit{sub-wormhole} $\sigma \subset \Sigma$. As a result, $C(\Sigma)$ should not be chosen merely to cancel the contribution of $\Sigma$ to the gravitational path integral, but only the part of its contribution which cannot be ascribed to any sub-wormhole.\footnote{If our theory were not topological, we might want to especially focus on the contribution of metrics on $\Sigma$ where the sub-wormhole $\sigma \subset \Sigma$ is becoming very small. An interesting point of comparison is the definition of $n$-point interaction vertices in closed string field theory: the $n$-point vertex is not merely the $n$-point string amplitude, but involves subtracting off some part corresponding to the sewing of lower-point interaction vertices with propagators arising from a neighborhood of the degeneration limits in moduli space \cite{Zwiebach:1992ie}.}

To handle the possibility of sub-wormholes, we would like to construct a gravitational analog of the BPHZ formula \eqref{eq:renorm_value}. We first need to determine what the analog of the unrenormalized Feynman integrals $U(\Gamma)$ are. These will be given by \textit{unrenormalized wormhole factors} $U(\Sigma)$, defined to be equal to $1$ for all $\Sigma$, so that we may rewrite the Marolf--Maxfield path integral \eqref{eq:grav_tree_and_loops} as
\begin{equation}\label{eq:MM_with_U}
    \braket{Z^n} = \lambda^n + \sum_{\Sigma \in \mathcal{W}_n^{\rm wh}} \lambda^{k_\Sigma} U(\Sigma).
\end{equation}
With the somewhat trivial functional $U(\Sigma)$ defined, we can now define a \textit{prepared wormhole factor} for a connected surface $\Sigma$ via an analog of \eqref{eq:prepared_value}
\begin{equation}\label{eq:grav_prepared_value}
\overline{R}(\Sigma) = U(\Sigma) + \sum_{ \sigma\subset\Sigma} C(\sigma) U(\Sigma/\sigma),
\end{equation}
where $\Sigma/\sigma$ denotes the surface obtained by excising $\sigma$ from $\Sigma$ and gluing in disks along the new boundaries (see Figure \ref{fig:excision}). 

\begin{figure}
\centering
\includegraphics[scale=0.5]{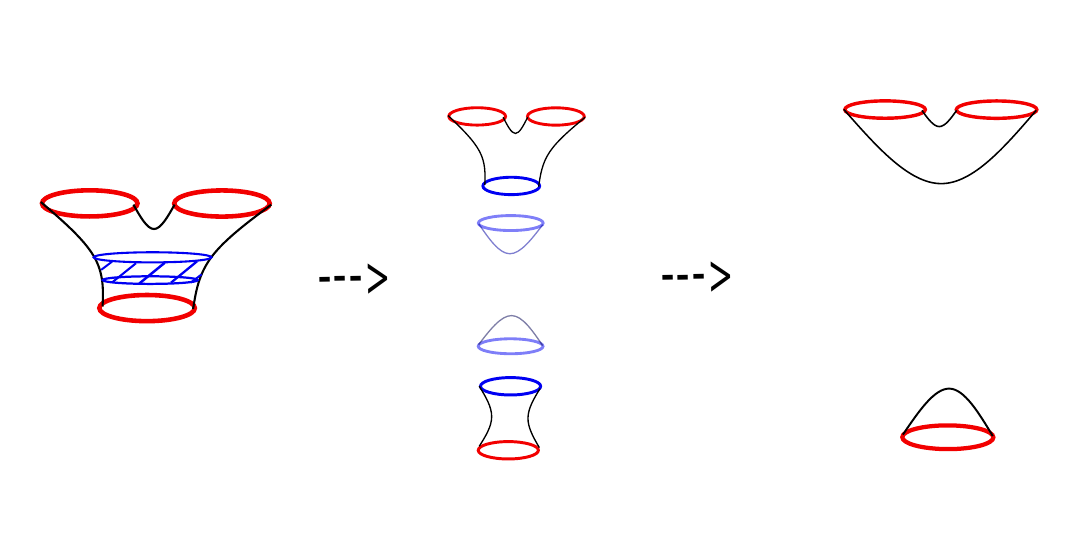}
\caption{The construction of the manifold $\Sigma/\sigma$, when $\Sigma$ is a pair of pants and $\sigma$ an embedded cylinder. We excise the embedded cylinder and fill each hole with a disk, so that we are left with the disjoint union of a disk and a cylinder.}
\label{fig:excision}
\end{figure}

We extend $\overline{R}(\Sigma)$ to disconnected $\Sigma$ multiplicatively. The intuition behind \eqref{eq:grav_prepared_value} is that for each sub-wormhole $\sigma \subset \Sigma$, we subtract off the counter-wormhole factors associated to $\sigma$, in order to quantify the non-factorization arising from $\Sigma$ that cannot be ascribed to a sub-wormhole. A precise definition of the set of (proper, non-empty) sub-wormholes $\sigma \subset \Sigma$ over which we sum in \eqref{eq:grav_prepared_value} is as follows.

\begin{defn}\label{defn:subworm_set}
    Let $\Sigma_n$ denote a connected genus zero surface with $n$ circular boundaries. The set of proper, non-empty sub-wormholes over which we sum in \eqref{eq:grav_prepared_value} is the set of (mapping class group equivalence classes of) connected\footnote{The requirement that the sub-wormholes $\sigma \subset \Sigma_n$ be connected is a bit artificial, and is chosen only to simplify the algebraic structure. As with the case of higher-genus surfaces, see \cite{Blommaert:2021fob} for a treatment that does not make this restriction.} embedded submanifolds $\sigma \subset \Sigma_n$ such that $\sigma \cong \Sigma_k$ for some $2 \leq k < n$ and such that the surface $\Sigma_n/\sigma \cong \Sigma_{n_1} \sqcup \cdots \sqcup \Sigma_{n_k}$ for some partition $n = n_1 + \cdots + n_k$ of the set of boundaries of $\Sigma_n$ into $k$ non-empty parts. Note that a sub-wormhole $\sigma \subset \Sigma_n$ is uniquely specified (up to the action of the mapping class group) by the partition it induces on the set of $n$ boundaries.
\end{defn}

With the prepared wormhole factors $\overline{R}(\Sigma)$ defined, we now define the counter-wormhole factors $C(\Sigma)$ recursively by $C(\Sigma) = - \overline{R}(\Sigma)$ for all wormholes $\Sigma$, so that the \textit{renormalized wormhole factors}
\begin{equation}
R(\Sigma)=C(\Sigma) + \overline{R}(\Sigma)=C(\Sigma)+U(\Sigma)+\sum_{\sigma\subset\Sigma}C(\sigma)U(\Sigma/\sigma)
\label{eq:renormgrav}
\end{equation}
vanish for all wormholes, $R(\Sigma) = 0$. As $\overline{R}(\Sigma_n)$ only involves wormhole factors $C(\Sigma_k)$ for $k < n$, this recursive process terminates. We will refer to this recursive definition of $C(\Sigma)$ as the \textit{gravitational BPHZ procedure}.

Now that we have defined the counter-wormhole factors $C(\Sigma)$ for all wormholes $\Sigma$, we can use them to construct a renormalized, factorizing gravitational path integral in three different ways, directly analogous to the three approaches to defining renormalized perturbation theory described in Section \ref{sec:SUBDIV_BPHZ}:

\textbf{First approach:} The first, most direct approach is to replace the unrenormalized wormhole factors $U(\Sigma)$ with the renormalized ones $R(\Sigma)$ in the gravitational path integral \eqref{eq:MM_with_U}, obtaining
\begin{equation}\label{eq:first_approach_grav}
    \braket{Z^n} = \lambda^n + \sum_{\Sigma \in \mathcal{W}_n^{\rm wh}} \lambda^{k_\Sigma} R(\Sigma).
\end{equation}
Since we have set $R(\Sigma) = 0$ for all wormholes $\Sigma$, this reduces to removing all wormholes by hand in the sum over topologies, and we are left with
\begin{equation}
    \braket{Z^n} = \lambda^n,
\end{equation}
a factorizing result. While somewhat tautological, excluding wormholes by hand does have the benefit of making factorization manifest, just as the first approach to defining renormalized perturbation theory in Section \ref{sec:SUBDIV_BPHZ} makes the cancelation of UV divergences manifest. However, this approach involves dramatically changing the rules of the game by hand, which obscures the meaning of the many useful semiclassical computations that include summing over wormholes.

\textbf{Second approach:} The second approach is to use $C(\Sigma)$ to define a non-local bare action functional $e^{-S_{\rm bare}(\Sigma)}$ as in \eqref{eq:counterwormhole_ansatz}. We then forget where this bare action functional came from, and compute using the unrenormalized wormhole factors $U(\Sigma) = 1$, but now with the contribution of a surface $\Sigma$ being given by $e^{-S_{\rm bare}(\Sigma)}$ rather than $\lambda^{k_{\Sigma}}$, obtaining \eqref{eq:modified_grav_path_integral}, which we can rewrite here as
\begin{equation}\label{eq:second_approach_grav}
    \braket{Z^n} = e^{- S_{\rm bare}(\text{$n$ disks})} + \sum_{\Sigma \in \mathcal{W}_n^{\rm wh}} e^{-S_{\rm bare}(\Sigma)} U(\Sigma).
\end{equation}
While it may not be obvious yet, if the counter-wormhole factors $C(\Sigma)$ are chosen recursively through the gravitational BPHZ procedure, then this second approach will agree with the first, and yield $\braket{Z^n} = \lambda^n$, as desired. In this approach, we perform a nontrivial sum over wormhole contributions, but the non-localities in the bare action functional will conspire to cancel against the wormholes and lead to factorization. While this approach obscures factorization, it makes manifest that we are still performing a sum over the same topologies as in the semiclassical Marolf--Maxfield theory, with non-local rules for the contribution of a spacetime depending on its number of connected components. It is worth pointing out that the non-localities in $e^{-S_{\rm bare}(\Sigma)}$ are all subleading in $\lambda^{-1}$, so we may view these as non-local corrections to the Marolf--Maxfield EFT, suppressed by $\lambda^{-1}$.

\textbf{Third approach:} The third approach consists of enlarging the set of topologies we include in our gravitational path integral, by explicitly including counter-wormhole topologies. In this approach, each wormhole or counter-wormhole contribution will still be individually non-factorizing, but the combinatorics of formula \eqref{eq:renormgrav} will be more readily visible, since we can keep track of the contributions of each counter-wormhole individually.

\begin{figure}
\centering
\includegraphics[scale=0.5]{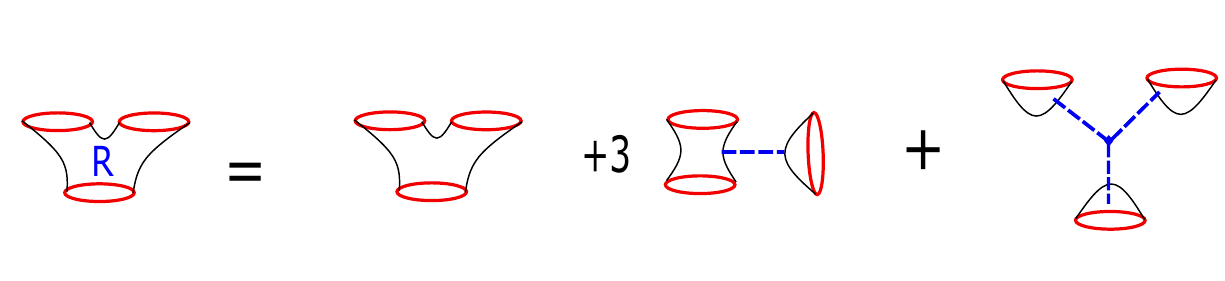}
\caption{A graphical representation of the counter-wormholes contributing to the renormalized value of a pair of pants. In order to restore factorization, we derive the contributions of the counter-wormholes recursively by asking that the renormalized values of all nontrivial wormholes are zero.}
\label{fig:Renormalizedpants}
\end{figure}

\subsection{A few explicit examples}
\label{sec:feworders}
We now show how the gravitational BPHZ procedure implemented by \eqref{eq:renormgrav} implements factorization for the first few correlation functions $\braket{Z^2}, \braket{Z^3}$. In addition, we will show explicitly that the three approaches for incorporating the counter-wormhole factors $C(\Sigma)$ discussed above are equivalent.

Before being able to apply any of these approaches, we first need to compute the factors $C(\Sigma_2), C(\Sigma_3)$ associated to the cylinder $\Sigma_2$ and the pair of pants $\Sigma_3$. For the cylinder $\Sigma_2$, there are two terms in Equation \eqref{eq:renormgrav}, which reads
\begin{equation}
R(\Sigma_2) = C(\Sigma_2) + U(\Sigma_2).
\end{equation}
Since we want to impose $R(\Sigma_2)=0$, and we have that $U(\Sigma_2)=1$, we learn that
\begin{equation}
C(\Sigma_2)=-1.
\end{equation}
In other words, the counter-wormhole associated to the cylinder exactly compensates the contribution of the cylinder, as there are no nontrivial sub-wormholes to consider.

Let us move on to the next order. For the pair of pants $\Sigma_3$, Equation \eqref{eq:renormgrav} reads
\begin{equation}
R(\Sigma_3) = C(\Sigma_3) + U(\Sigma_3) + 3 C(\Sigma_2) U(\Sigma_2 \sqcup \Sigma_1),
\end{equation}
arising from the three topologically distinct sub-wormholes $\Sigma_2 \subset \Sigma_3$ associated to the three partitions of the set $\{1, 2, 3\}$ into two non-empty parts. By imposing $R(\Sigma_3)=0$ and inserting the previously computed value $C(\Sigma_2) = -1$, we obtain:
\begin{equation}
C(\Sigma_3) = - U(\Sigma_3) - 3 C(\Sigma_2) U(\Sigma_2 \sqcup \Sigma_1) = - 1 - 3 (-1) = 2.
\end{equation}
This last computation is illustrated on Figure \ref{fig:Renormalizedpants}.

Let us now implement the three methods described above to these first two steps of the gravitational BPHZ procedure.

\textbf{First approach:} The above computations of the counterterms guarantee that $R(\Sigma_2)=R(\Sigma_3)=0$. Therefore all wormholes become excluded from the gravitational path integral by hand, and we simply calculate \begin{align}\braket{Z^2}=\lambda^2,\quad\quad \braket{Z^3}=\lambda^3,\end{align}which is of course consistent with factorization.

\textbf{Second approach:} We first compute $e^{-S_{\rm bare}(\Sigma)}$ for $k_\Sigma = 2, 3$, using \eqref{eq:counterwormhole_ansatz}. For $k_{\Sigma} = 2$ we obtain:
\begin{equation}
    e^{-S_{\rm bare}(\text{$2$ disks})} = \lambda^2 + \lambda^{k_{\Sigma_2}} C(\Sigma) = \lambda^2 - \lambda,
\end{equation}
and for $k_\Sigma = 3$ we obtain
\begin{equation}
    e^{-S_{\rm bare}(\text{$3$ disks})} = \lambda^3 + 3 \lambda^{k_{\Sigma_2 \sqcup \Sigma_1}} C(\Sigma_2) + \lambda^{k_{\Sigma_3}} C(\Sigma_3) = \lambda^3 - 3 \lambda^2 + 2 \lambda.
\end{equation}
As an explicit illustration, the computation of $e^{-S_{\rm bare}(\text{$3$ disks})}$ is represented on Figure \ref{fig:sbare}.

\begin{figure}
\centering
\includegraphics[scale=0.7]{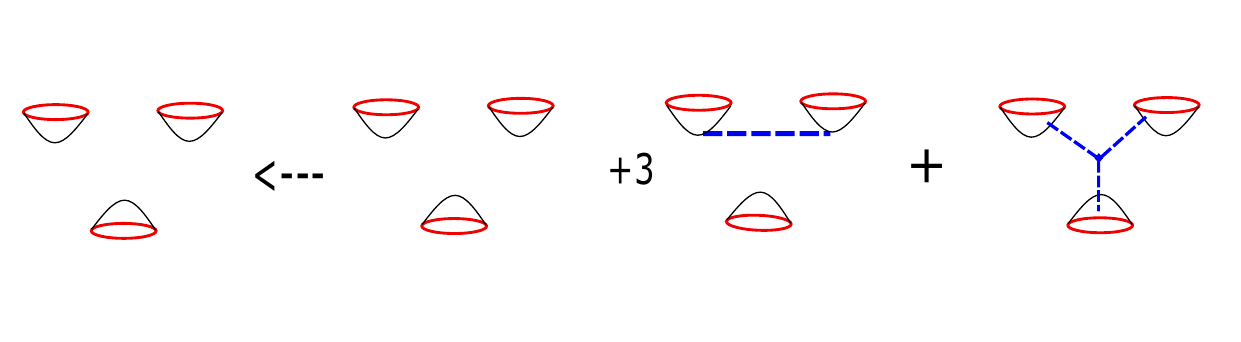}
\caption{A graphical representation of all the counter-wormholes getting resummed into the contribution $e^{-S_{\rm bare}(\text{$3$ disks})}$ of a spacetime with three connected components (here taken to be three disks). After resumming the counter-wormholes, this contribution is no longer $\lambda^3$: the counter-wormholes introduce non-localities.}
\label{fig:sbare}
\end{figure}

Now we can compute the values of $\braket{Z^2}$ and $\braket{Z^3}$ using Equation \eqref{eq:modified_grav_path_integral}:
\begin{align}
\braket{Z^2} = e^{-S_{\rm bare}(\Sigma_1 \sqcup \Sigma_1)} + e^{-S_{\rm bare}(\Sigma_2)} = (\lambda^2-\lambda) + \lambda = \lambda^2,
\label{eq:expanded2}
\end{align}
and
\begin{align}
\braket{Z^3} &= e^{-S_{\rm bare}(\Sigma_1 \sqcup \Sigma_1 \sqcup \Sigma_1)} + 3 e^{-S_{\rm bare}(\Sigma_2 \sqcup \Sigma_1)} + e^{-S_{\rm bare}(\Sigma_3)} \\
&= (\lambda^3 - 3 \lambda^2 + 2 \lambda) + 3 (\lambda^2 - \lambda) + \lambda = \lambda^3,
\label{eq:expanded3}
\end{align}
which explicitly shows that once again, factorization is restored.

\textbf{Third approach:} In this method, we represent each term of the expanded forms of Equations \eqref{eq:expanded2} and \eqref{eq:expanded3} as a topology of its own, including counter-wormholes explicitly, as shown on Figures \ref{fig:z2counterterm} and \ref{fig:Seff}. This corresponds to adding more topologies to the gravitational path integral, which make it explicit that we have included additional wormhole-like contributions to the gravitational path integral.
\begin{figure}
\centering
\includegraphics[scale=0.6]{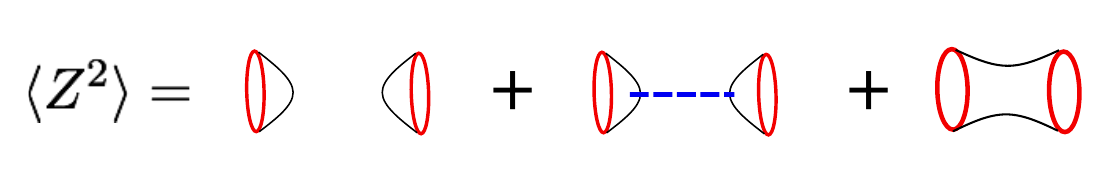}
\caption{The systematic expansion of the gravitational path integral with two boundaries, including the counterterm $C(\Sigma_2)$ represented here with a dashed line.}
\label{fig:z2counterterm}
\end{figure}
\begin{figure}
\centering
\includegraphics[scale=0.7]{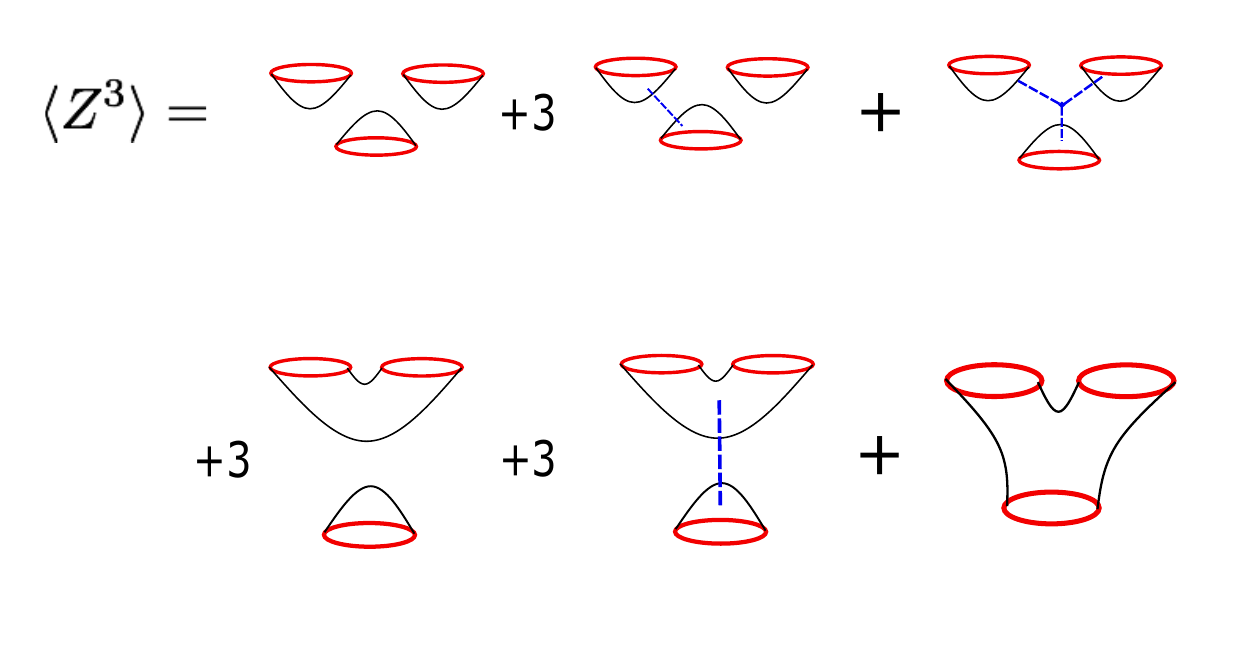}
\caption{The sum of all wormholes and counter-wormholes contributing to the spacetime amplitude $\braket{Z^3}$. The way we constructed counter-wormhole contributions recursively ensures that this sum leads to a factorizing answer.}
\label{fig:Seff}
\end{figure}

\subsection{The Faà di Bruno Hopf Algebra}
\label{sec:FAA_DI_BRUNO}
Now that we have gained some intuition on how the gravitational BPHZ procedure leads to factorization for the first few multi-boundary correlation functions, we are ready to describe the abstract algebraic structure controlling this recursive process. We will define a Hopf algebra $\mathcal{H}_{\rm FdB}$ that underlies the gravitational BPHZ algorithm, analogous to the Connes-Kreimer Hopf algebra $\mathcal{H}_{\rm CK}$ that underlies the BPHZ algorithm in perturbative renormalization. In addition, we will be able to interpret the group of characters $G_{\rm FdB}$ of this Hopf algebra as a group of gauge equivalences or dualities which interpolate between equivalent prescriptions for the gravitational path integral, and in particular, between the first, second, and third approaches described above.

In the case of perturbative renormalization, we saw that many transformations of the form \eqref{eq:CK_action}. In the case of gravity, many of the key equations (such as \eqref{eq:counterwormhole_ansatz}, \eqref{eq:MM_with_U}, \eqref{eq:first_approach_grav}, or \eqref{eq:second_approach_grav}) were of the form 
\begin{equation}\label{eq:FdB_action_grav}
    B_n = A_n + \sum_{\Sigma \in \mathcal{W}_n^{\rm wh}} A_{k_\Sigma} F(\Sigma).,
\end{equation}
where $A_n, B_n$ are sequences of values assigned to spacetimes with $n \geq 1$ components (for instance, $\lambda^n, B_n(\lambda)$, or $e^{-S_{\rm bare}(\text{$n$ disks})}$), and $F(\Sigma)$ some functional on the set of wormholes that is multiplicative on disjoint connected components.

We wish to interpret the transformation \eqref{eq:FdB_action_grav} as the right action of a group $G_{\mathrm{FdB}}$. Functionals $F(\Sigma)$ will define elemnts of $G_{\mathrm{FdB}}$, and we aim to summarize the discussion in Section \ref{sec:RENORM_MM} in terms of a commutative diagram 
\begin{equation}\label{eq:two_approaches_grav}
   \begin{tikzcd}[column sep=huge]
    \lambda^n \arrow[r, "C"] \arrow[dr, "R"'] \arrow[d, "U"'] & e^{-S_{\rm bare}(\text{$n$ disks})} \arrow[d, "U"] \\
    \braket{Z^n} = B_n(\lambda) & \braket{Z^n} = \lambda^n
    \end{tikzcd}
\end{equation}
In this diagram, the vertical arrows $U$ use a sequence $A_n$ to compute an unrenormalized gravitational path integral
\begin{equation}
\braket{Z^n} = A_n + \sum_{\Sigma \in \mathcal{W}_n^{\rm wh}} A_{k_\Sigma},
\end{equation}
the diagonal arrow $R$ computes a tautological ``renormalized gravitational path integral''
\begin{equation}
\braket{Z^n} = A_n + \sum_{\Sigma \in \mathcal{W}_n^{\rm wh}} A_{k_\Sigma} R(\Sigma) = A_n,
\end{equation}
and the horizontal arrow $C$ sums over counter-wormhole contributions according to the ansatz \eqref{eq:counterwormhole_ansatz}. The three approaches to the gravitational BPHZ procedure correspond to different ways of reaching the lower right hand corner, and result in equivalent physical theories.

As in the case of Connes-Kreimer, we can motivate the group law $\star$ on $G_{\rm FdB}$ by observing that we should have $R = C \star U$, and extrapolating from \eqref{eq:renormgrav} to define:
\begin{equation}\label{eq:FdB_group_law}
(F\star G)(\Sigma)=F(\Sigma)+G(\Sigma)+\sum_{\sigma\subset\Sigma}F(\sigma)G(\Sigma/\sigma),
\end{equation}
where $\sigma$ runs over the set of sub-wormholes $\sigma \subset \Sigma$ specified in Definition \ref{defn:subworm_set}. Dually, we can define the Hopf algebra $\mathcal{H}_{\rm FdB}$ dual to $G_{FdB}$, which turns out to be isomorphic to a well-known Hopf algebra called the Faà di Bruno Hopf algebra (as proven below in Proposition \ref{prop:equiv_to_FdB}):

\begin{defn}
    The \textit{Faà di Bruno algebra} $\mathcal{H}_{\rm FdB}$ is the free commutative algebra $\mathbb{C}[\Sigma_2, \Sigma_3, \dots]$ generated by the set of (diffeomorphism classes of) connected genus zero surfaces $\Sigma$ with $n \geq 2$ circular boundaries, equipped with counit $\varepsilon : \mathcal{H}_{\rm FdB} \to \mathbb{C}$ and coproduct $\Delta : \mathcal{H}_{\rm FdB} \to \mathcal{H}_{\rm FdB} \otimes \mathcal{H}_{\rm FdB}$ defined on algebra generators $\Sigma$ by $\varepsilon(\Sigma) = 0$ and
    \begin{equation}
        \Delta(\Sigma) = \Sigma \otimes 1 + 1 \otimes \Sigma + \sum_{\sigma\subset\Sigma} \sigma \otimes \Sigma/\sigma,\label{eq:gravcoprod}
    \end{equation}
    and extended multiplicatively. The sum in \eqref{eq:gravcoprod} runs over all sub-wormholes $\sigma \subset \Sigma$ as specified in Definition \ref{defn:subworm_set} (see Figure \ref{fig:CoprodPants}), and we identify the disjoint union of surfaces with the free algebra structure on $\mathcal{H}_{\rm FdB}$. We further define $G_{\rm FdB}$ to be the spectrum ${\rm Spec}(\mathcal{H}_{\rm FdB})$, consisting of algebra homomorphisms $F : \mathcal{H}_{\mathrm{FdB}} \to \mathbb{C}$. The coproduct $\Delta$ defines a multiplication law $\star : G_{\rm FdB} \times G_{\rm FdB} \to G_{\rm FdB}$ given by \eqref{eq:FdB_group_law}.
\end{defn}

\begin{figure}
\centering
\includegraphics[scale=0.8]{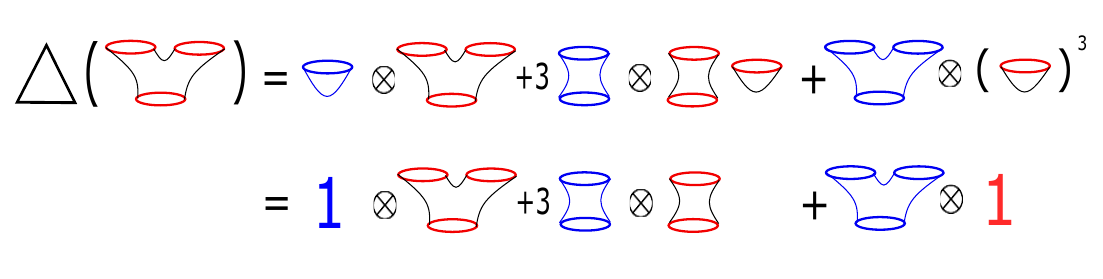}
\caption{The coproduct of a pair of pants wormhole. In the second line we have made it manifest that disks get identified with the unit of the Faà di Bruno algebra.}
\label{fig:CoprodPants}
\end{figure}

\begin{prop}
The coproduct $\Delta$ equips $\mathcal{H}_{\mathrm{FdB}}$ with the structure of a Hopf algebra.
\end{prop}

\begin{proof}
This proof could be circumvented by using the following Proposition \ref{prop:equiv_to_FdB} that $G_{\mathrm{FdB}}$ is isomorphic to a well-known group. Nevertheless, we present a more direct proof along the lines of the proof of Proposition \ref{prop:CK_is_Hopf} because it is valuable for intuition.

By definition, $\varepsilon$ and $\Delta$ are algebra homomorphisms, and it is immediate that $\varepsilon$ is a two sided counit for $\Delta$. For coassociativity, observe that for connected surface $\Sigma$, both $(\Delta \otimes {\rm Id})(\Delta(\Sigma))$ and $({\rm Id} \otimes \Delta)(\Delta(\Sigma))$ can be written as
\begin{equation}
\Delta^2(\Sigma)=\sum_{\sigma\subset\sigma^\prime\subset\Sigma}\sigma\otimes\sigma^\prime/\sigma\otimes\Sigma/\sigma^\prime,
\end{equation}
where the sum runs of nested subwormholes $\sigma \subset \sigma' \subset \Sigma$ and further includes empty and improper subwormholes, where we identify the disk $\Sigma_1$ with the unit $1 \in \mathcal{H}_{\rm FdB}$.

This shows that $\mathcal{H}_{\mathrm{FdB}}$ is a bialgebra. Equipped with the grading $|\Sigma_n| = n - 1$, we see that $\mathcal{H}_{\mathrm{FdB}}$ is connected graded, so it a Hopf algebra by Proposition \ref{prop:conn_graded}, with the antipode defined recursively through the formula
\begin{align}
S(\Sigma)=1-\Sigma-\sum_{\sigma\subset\Sigma}S(\sigma) \Sigma/\sigma.
\label{eq:grav_antipode}
\end{align}
\end{proof}

\begin{cor}
$(G_{\mathrm{FdB}},\star)$ is a group.
\end{cor}
\begin{proof}
    This follows from Proposition \ref{prop:dual_group_to_Hopf_algebra}.
\end{proof}

We now introduce another presentation of $G_{\mathrm{FdB}}$ that makes some of its algebraic structure more manifest, and will greatly simplify the calculations in a lot of cases. Indeed, while the presentations of the Faà di Bruno Hopf algebra and its associated group given above are the most natural ones from the point of view of the gravitational theories at hand, they are isomorphic to objects that may look more familiar to the reader. In particular, the group $G_{\mathrm{FdB}}$ can be identified with the group of formal power series with zero constant term and unit linear term, equipped with a much more familiar law: the law $\circ$ of composition. More precisely, we have:

\begin{prop}\label{prop:equiv_to_FdB}
The map
\begin{align}
\begin{array}{cccc}
    (G_{\mathrm{FdB}},\star) &\longrightarrow  & (t+t^2\mathbb{C}[[t]],\circ)  \\
     F &\longmapsto &\hat{F}(t)=t+\sum_{k\geq 2}F(\Sigma_k) \frac{t^k}{k!},
\end{array}
\end{align}
where $\Sigma_k$ denotes the genus zero topological surface with $k$ boundaries, is an isomorphism of groups.
\end{prop}
\begin{proof}
Let $F,G\in G_{\mathrm{FdB}}$. We have that\begin{align}\widehat{(F\star G)}^{(n)}(t)=(F\star G)(\Sigma_n) =\sum_{k_1+\dots+k_p=n} F(\Sigma_p) \left(\prod_{i=1}^p G(\Sigma_{k_i})\right)=(\hat{F}\circ \hat{G})^{(n)}(t),\end{align}
by Faà di Bruno's formula for the $n^{th}$ derivative of the composition of two formal power series. Thus, $F \mapsto \hat F$ is a group homomorphism. Moreover it admits an inverse, defined by setting $F(\Sigma_k) = \hat{F}^{(k)}$, so it defines an isomorphism of groups.
\end{proof}

\begin{rem}
Proposition \ref{prop:equiv_to_FdB} justifies the name of the Faà di Bruno Hopf algebra: the combinatorics underlying our gravitational path integral are the same as the ones underlying the composition of formal power series (as is already evident in \cite{MaMax}), which leads to the Faà di Bruno formula for the $n^{th}$ derivative of a composite function. Both cases involve a sum over partitions of $n$, which in our case correspond to all the possible wormhole configurations between $n$ boundaries.
\end{rem}
\begin{rem}
The inverse law of the group $t+t^2\mathbb{C}[[t]]$ under composition is given by the well-known Lagrange inversion formula:
\begin{align}
\hat{F}^{-1}(t)=\sum_{n=1}^\infty\frac{1}{n!}\frac{\mathrm{d}^{n-1}}{\mathrm{d}s^{n-1}}\left(\frac{\hat{F}(s)}{s} \right)^{-n}\bigg\vert_{s=0}t^n.
\end{align}
\end{rem}

To define the action of $G_{\rm FdB}$ on the space of possible gravitational action functionals, we make the following definition.

\begin{defn}
    The space of \textit{bulk action functionals} is the space of formal power series
    \begin{equation}
        A(t) = 1 + \sum_{k > 0} A_k \frac{t^k}{k!},
    \end{equation}
    with constant term $1$. We interpret the power series $A(t)$ as the gravitational action functional that assigns the value $e^{- S(\Sigma)} = A_{k_\Sigma}$ to any spacetime manifold $\Sigma$ with $k_\Sigma$ connected components.
\end{defn}

\begin{exmp}
    The bulk action functional for Marolf--Maxfield theory is
    \begin{equation}\label{eq:power_series_MM}
        A_{\rm MM}(t) = \sum_{k \geq 0} \lambda^k \frac{t^k}{k!} = e^{\lambda t},
    \end{equation}
    while the effective bulk action functional after integrating out wormholes is
    \begin{equation}
        A_{\rm MM}^{\rm eff}(t) = \sum_{k \geq 0} B_k(\lambda) \frac{t^k}{k!} = e^{\lambda (e^t - 1)},
    \end{equation}
    to be compared with \cite[Section 3.2]{MaMax}.
\end{exmp}

\begin{prop}
Equation \eqref{eq:FdB_action_grav} defines a right action of the group $G_{\mathrm{FdB}}$ on the set of bulk action functionals $A(t) = 1 + \sum_{k \geq 1} A_n t^k/k!$.
\end{prop}

\begin{proof}
For $F\in G_{\mathrm{FdB}}$, define the right action
\begin{equation}\label{eq:FdB_action_as_power_series}
    (A\cdot F)(t) = A(t) + \sum_{\Sigma \in \mathcal{W}^{\rm wh}} A_{k_\Sigma} F(\Sigma)\frac{t^{n_\Sigma}}{(n_\Sigma)!},
\end{equation}
where the sum runs over the set $\mathcal{W}^{\rm wh} = \coprod_n \mathcal{W}_n^{\rm wh}$ of wormhole spacetimes $\Sigma$ with any number of boundaries $n_\Sigma$. Note that the $n$th component $(A \cdot F)^{(n)}(t)$ of \eqref{eq:FdB_action_as_power_series} is precisely given by \eqref{eq:FdB_action_grav}. Applying Faà di Bruno's formula to \eqref{eq:FdB_action_as_power_series}, we have that
\begin{align}
(A\cdot F)(t)= (A\circ \hat{F})(t),
\end{align}
using the isomorphism constructed in Proposition \ref{prop:equiv_to_FdB}. Acting by precomposition is manifestly a right action.
\end{proof}

\subsection{Factorization at all orders}\label{sec:allorders}

Leveraging the algebraic framework developed in the previous section, we can now compute appropriate counter-wormhole factors $C(\Sigma_n)$ to restore factorization in the Marolf--Maxfield path integral at all orders. In particular, the constraint that $R(\Sigma) = 0$ for all $\Sigma$ ensures that $R = \varepsilon$, and so $R$ is simply the unit of $G_{\rm FdB}$. Thus, the equation $R = C \star U$ tells us that $C$ is simply given by the inverse of $U$ under the group law: we have $C = U \circ S$, where $S$ is the antipode on $\mathcal{H}_{\rm FdB}$.

The easiest way to compute $C$ is through the isomorphic presentation of $G_{\rm FdB}$ as a group of formal power series under composition provided by Proposition \ref{prop:equiv_to_FdB}. Note that the image of $U$ under this isomorphism is the formal power series
\begin{equation}
    \hat{U}(t) = t + \sum_{k \geq 2} \frac{t^k}{k!} = e^t - 1.
\end{equation}
Thus, it is straightforward to compute $\hat{C}(t)$ as the compositional inverse of $\hat{U}(t)$, given by
\begin{equation}
    \hat{C}(t) = \log(1 + t).
\end{equation} By Taylor expanding this expression, we obtain:
\begin{prop}
The counter-wormhole factors for Marolf--Maxfield theory to all orders are given by
\begin{equation}
C(\Sigma_n)=(-1)^{n+1}(n-1)!
\end{equation}
\end{prop}

We can now recast the three approaches to defining a renormalized gravitational path integral in terms of the identity
\begin{equation}
    e^{\lambda t} = (A_{\rm MM} \cdot C \cdot U)(t),\label{eq:gauge}
\end{equation}
where the left hand side is the exponential generating function of the factorizing multi-boundary correlation functions $\braket{Z^n} = \lambda^n$. More precisely, the three approaches to incorporating counter-wormholes differ in the order in which we perform the operations and the subsequent simplifications performed.

\textbf{First approach:} First perform the group multiplication $R=C\star U$, so that we obtain \begin{align}e^{\lambda t} = (A_{\rm MM}\cdot R)(t),\end{align} where $\hat{R}(t)=t$. This amounts to giving a zero contribution to all spacetimes with wormholes inside the gravitational path integral, and simply reading off spacetime amplitudes by evaluating the Marolf--Maxfield action on only disconnected spacetimes.

\textbf{Second approach:} First act with $C$ to define
\begin{equation}
    A_{\rm MM}^{\rm bare}(t) = (A_{\rm MM}\cdot C)(t) = 1 + \sum_{k \geq 1} e^{- S_{\rm bare}(\text{$k$ disks})} \frac{t^k}{k!}.
\end{equation}
to define a bare action, leaving us with
\begin{equation}
    e^{\lambda t} = (A_{\rm MM}^{\rm bare} \cdot U)(t).
\end{equation}
This approach corresponds to resumming the counter-wormholes into non-local spacetime actions $e^{- S_{\rm bare}(\Sigma)}$ for the contribution of a spacetime $\Sigma$ with multiple connected components, and then doing an ordinary sum over topologies with this new action functional. In particular, we can easily compute
\begin{equation}\label{eq:bareMM}
    A_{\rm MM}^{\rm bare}(t) = (A_{\rm MM} \circ \hat{C})(t)=(1+t)^\lambda.
\end{equation}

Expanding in powers of $t$ and extracting coefficients, we deduce that a factorizing bare action functional for Marolf--Maxfield theory is given by
\begin{equation}\label{eq:MM_bare_action}
    e^{- S_{\rm bare}(\Sigma)} = \lambda(\lambda-1)\dots(\lambda-k_\Sigma+1).
\end{equation}
As described in \cite{MaMax}, something special happens when $\lambda \in\mathbb{N}$ : the Taylor expansion of \eqref{eq:bareMM} terminates, and $e^{-S_{\mathrm{bare}}(\Sigma_n)}=0$ whenever $n>\lambda$. This means that the sum over counter-wormholes for Marolf--Maxfield theory for $\lambda \in \mathbb{N}$ enforces the non-local constraint that only topologies with $\lambda$ connected components or fewer may contribute! The difference between this description and the one given in \cite{MaMax} for an $\alpha$-sector with $\alpha = \lambda \in \mathbb{N}$ is that they fix the number of connected components to be exactly $\alpha$, while we have not allowed for closed universes so can only set an upper bound on the number of connected components. In our presentation, there is nothing to stop us from taking $\lambda \notin \mathbb{N}$, as we have not imposed reflection positivity or the existence of a well defined Hilbert space.

The formula \eqref{eq:MM_bare_action} also has the interpretation as a gravitational ``exclusion rule" as described in \cite{Saad:2021rcu,Saad:2021uzi}. To see this, imagine that in the first approach (where wormholes are excluded by hand), the path integral counts $\lambda$ many microstates of the bulk quantum gravity theory on each disk (which could be interpreted as states of half-wormholes or end-of-the-world branes). Then, to reintroduce wormholes, we take any configuration where multiple connected components are in the same microstate, and glue them into a connected topology (this is the ``diagonal = wormhole'' principle of \cite{Saad:2021rcu,Saad:2021uzi}). Thus, if one takes this second approach, and sums both disconnected topologies and wormholes in the gravitational path integral, states where distinct connected components lie in the same microstate need to be subtracted by hand in order to avoid double counting, leading to the formula \eqref{eq:MM_bare_action}.

\textbf{Third approach:} Perform none of the compositions in \eqref{eq:gauge}, and merely expand out all three terms $A_{\rm MM}, C,$ and $U$. This amounts to adding more topologies into the gravitational path integral, which corresponds to all the counter-wormholes being explicitly spelled out.

\begin{rem}
The identity \eqref{eq:gauge} makes it very clear that the Faà di Bruno group $G_{\rm FdB}$ has an interpretation as a group of gauge transformations (or dualities) involving topology change inside the gravitational path integral. Indeed, we can choose arbitrary elements $\theta_1,\theta_2\in G_{\rm FdB}$ and insert a product $1 = \theta_i^{-1} \theta_i$ in between two operations of \eqref{eq:gauge}. We obtain:
\begin{align}\label{eq:grav_gauge_redundancy}
e^{\lambda t} = (A_{\rm MM}\cdot \theta_1^{-1}\cdot \theta_1 \cdot C\cdot\theta_2^{-1}\cdot \theta_2\cdot U)(t).
\end{align}
The interpretation of $\theta_1$ is that it trades counter-wormhole contributions for non-localities in the rules for sums over topologies. This can be understood as an instance of the phenemenon described by Coleman, Giddings, and Strominger \cite{Col,GS1,GS2}, whereby we integrate out wormholes into non-local effects. The difference is that here, we are integrating out not the geometric wormholes visible in the EFT, but instead integrating out microscopic, stringy wormholes as parametrized by the counter-wormholes. In contrast, the interpretation of $\theta_2$ is that it implements the cancelation of counter-wormholes against geometric wormholes, which can be viewed as an example of ER = EPR for spacetime wormholes \cite{Maldacena:2013xja, JMkitp}. In general, the gauge transformations \eqref{eq:grav_gauge_redundancy} allow us to interpolate between infinitely many equivalent ways of organizing the gravitational path integral, including but not limited to the three approaches described above. For example, choosing $\theta_1=C^{-1}$ allows us to switch between the second and third approaches.
\end{rem}

\section{Discussion}
\label{sec:DISC}

Despite the (extreme) simplicity of the toy model considered here, we believe that the structures outlined in this paper are likely to exist in some form in much more general gravitational path integrals. In this last section, we mention several natural possible extensions of our work, and reflect on the lessons that might be drawn from our analogy between the gravitational path integral and perturbative renormalization.

\subsection{More complicated toy models}
\label{sec:ext}

This work focused on the simplest possible example that allowed us to make a precise analogy: a topological theory of gravity on surfaces of genus zero with nonempty boundary, even simpler than the model of \cite{MaMax}. In this context, the Hopf algebra underlying the gauge redundancies of the sum over topologies is the Faà di Bruno Hopf algebra. It would obviously be interesting to consider more complicated sums over topologies. In particular, finding a similar algebraic structure to incorporate genus insertions and end of the world branes seems to be a natural extension, in order to formalize the full analysis of \cite{Blommaert:2021fob}. In higher dimension, an interesting case study could be to try to restore factorization in the tensor models of \cite{Belin:2023efa,Jafferis:2024jkb}.

The mathematics of more complicated algebraic objects associated to topology change is likely to also contain some extra features that are absent in our model. For example, our gravitational path integrals only contain finitely many terms. In renormalization, the sums over Feynman diagrams are usually infinite due to the possibility of nesting arbitrarily many loops, and this is one of the reasons why the Connes--Kreimer Hopf algebra has a richer structure than the Faà di Bruno Hopf algebra. Such infinite sums are usually handled through combinatorial Dyson--Schwinger equations (see for example \cite{Yea}), which have a rich interplay with the Hochschild cohomology of the algebra. We expect such structures to play a key role in more complete gravitational path integrals.

In our simple model, the renormalized wormhole factor $R(\Sigma)$ was given by the counit of the Hopf algebra, which sets $R(\Sigma) = 0$ for all wormholes $\Sigma$, so that the first approach to defining a renormalized path integral was to simply exclude all wormholes. In perturbative renormalization, we only aim to subtract off the divergent part of loop integrals, leaving behind a finite piece. We would like to understand whether there is a possibility for an analogous ``factorizing piece'' of the wormhole contribution, so that the first approach is less tautological. Mathematically, the way that the finite piece is extracted from the divergent loop integral is through a Rota--Baxter structure and the associated Birkhoff factorization \cite{Connes:1998qv,CM1}. It would be interesting to see whether a nontrivial Rota--Baxter structure could control the extraction of a factorizing piece of the wormhole contribution in a more complicated toy model.

\subsection{Geometric models and interplay with bulk EFT}

Another obvious generalization would be to go beyond topology and include geometry and dynamical fields in our bulk gravitational theory. In such a context, a crucial possibility is that the introduction of counter-wormholes and ordinary bulk EFT counterterms may not be completely independent. In particular, the effect of including counter-wormholes would likely be to introduce multi-local terms in our bulk effective action, which could be viewed as multi-local counterterms that supplement the ordinary local counterterms included to regulate bulk loop integrals.

In addition, a geometric theory contains a natural hierarchy of wormhole configurations, organized according to their geometric size. In our simple topological toy model, the integration over geometric wormholes was done in a single step. In more realistic, geometric theories, it might be more reasonable to consider integrating out all wormholes smaller than a given scale. In this way, the counter-wormholes at each scale would only be chosen to cancel against the remaining geometric wormholes larger than our cutoff, and so the counter-wormholes factors would be scale-dependent. This would lead to a natural modification of the rules of bulk effective field theory and renormalization group flow.

\subsection{Bootstrapping a UV complete theory}

On a more ambitious note, we might hope to find the gravitational analog of a UV complete QFT, wherein we are allowed to take the cutoff to be arbitrarily high. In a UV complete QFT, as opposed to an EFT, we no longer view the counterterms as placeholders for unknown UV physics, and instead view them as part of the fundamental definition of a complete theory. In the gravitational context, this would mean defining a UV complete, factorizing gravitational path integral by introducing appropriate counter-wormholes at arbitrarily high scales. These counter-wormholes would no longer be viewed as a substitute for the ``stringy'' physics that resolves factorization; they would be that stringy physics.

For one example of this idea, the giant graviton expansion of supersymmetric gauge theories \cite{Gaiotto:2021xce,Lee:2023iil} describes the giant graviton branes required in the bulk theory in order to correctly match the trace relations present at finite $N$. These giant graviton brane contributions have a similar algebraic structure to the counter-wormholes discussed in this paper, and it would be very interesting to this connection more precise.

\subsection{Factorization in higher codimension}
\label{sec:categ}

Another interesting question is that of factorization in higher codimension. The only observable quantities we considered in this paper were multi-boundary partition functions $\braket{Z^n}$, as we only considered manifolds with boundary and no corners. However, a theory of quantum gravity contains more data than just partition functions, which are associated to asymptotic boundary conditions of higher codimension. For example, it should be possible to associate Hilbert spaces to spatial slices of the boundary.\footnote{For a definition of these Hilbert spaces given only a reflection positive partition function, see for example \cite{Colafranceschi:2023txs,Colafranceschi:2023urj}.} The locality of the holographic dual requires not only that partition functions factorize, but that these Hilbert spaces (as well as higher-codimension data) factorize as well.\footnote{See e.g. \cite{Harlow:2018tqv, Colafranceschi:2023txs,Colafranceschi:2023urj, Boruch:2024kvv}. This has been dubbed the ``factorisation'' paradox \cite{Penington:2023dql} (note the spelling), as suggested by Henry Maxfield. We avoid this naming convention, as there are additional factorization paradoxes in higher codimension as well (for instance, the factorization of the category of boundary conditions), and we would quickly run out of letters!} In particular, factorization of Hilbert spaces is the original context in which gauge equivalences like ER = EPR wre originally considered, in order to explain how a spatial wormhole state could fit inside a tensor-factorized Hilbert space \cite{Maldacena:2001kr, Maldacena:2013xja}.

In order to incorporate factorization of Hilbert spaces into our setup, one would need to formulate an analog of the Hopf algebraic approach presented here for a space of topological surfaces with corners.\footnote{The Hilbert spaces associated to manifolds with corners in the toy model we consider were discussed in \cite{MaMax}. The existence of such a Hilbert space is one way to see that the partition function $\braket{Z}$ assigned to a single boundary circle must be a natural number, as it is the dimension of a finite-dimensional vector space.} We expect that the appropriate analog of a Hopf algebra in order to study factorization of Hilbert spaces would be an algebra of a higher-categorical nature, since the data of a Hilbert space can be understood as the categorification of the data of a partition function. It would be interesting to determine the appropriate higher-categorical algebra required to introduce appropriate counter-wormholes in order to guarantee factorization in all codimensions.

The resolution of these higher factorization paradoxes would also likely require the introduction of additional tools from algebraic topology. In particular, the Swampland Cobordism Conjecture \cite{McNamara:2019rup,MV} was originally introduced in order to address these failures of factorization \cite{McNamara:2022xkg}, and suggests that one should always be able to get rid of the sum over topologies in string theory by performing successive surgeries via ER = EPR \cite{JMkitp}. It seems promising to try to use the framework described in this paper to understand what such a procedure could look like.

One quick guess as to the possible answer is that the counter-wormholes needed to guarantee factorization in higher codimension would likely induce non-localities of a more general form, such as integrals over higher dimensional cycles as opposed to merely multi-local contact interactions. The reason is that these counter-wormholes would be needed to cancel against spacetime topologies with handles of higher dimension, which attach along higher dimensional cycles as opposed to a discrete set of points. These non-local contributions could potentially arise from the response of the worldvolume QFTs living on higher-dimensional branes which have been integrated out. These branes can be viewed as providing the microstates of the higher-dimensional handles via geometric transition.

\subsection{A cosmic Galois group for quantum gravity?}

As a final speculation, note that in this paper, note that we studied \textit{one} very simple gravitational theory with wormholes, and showed that the gravitational BPHZ procedure was described by a particular Hopf algebra $\mathcal{H}_{\rm FdB}$, as well as its dual group $G_{\rm FdB}$. In the case of perturbative renormalization, there is one Connes-Kreimer group $G_{\rm CK}$ per quantum field theory. However the story does not end there. Indeed, it was shown in \cite{Connes:2004zi} that there is a \textit{unique} theory-independent group $\mathbb{U}$ which acts on the space of couplings of \textit{all} QFTs and can be seen as a universal group of symmetries for perturbative quantum field theory. This group also allows to systematically implement the renormalization group flow \cite{Connes:2004zi}. The group $\mathbb{U}$ has very interesting number theoretic properties and is related to motivic Galois theory (see \cite{CM1} for a review). In reference to Pierre Cartier's conjecture that the renormalization group is closely related to the absolute Galois group $\mathrm{Gal}(\overline{\mathbb{Q}}/\mathbb{Q})$, the group $\mathbb{U}$ has been referred to as a ``cosmic Galois group."

Thus, it is tantalizing to ask whether a similar picture holds when it comes to the sum over topologies in quantum gravity. We put forward the following conjecture:

\begin{conj}
There exists a gravitational cosmic Galois group that underlies the gauge redundancies associated with topology change in all quantum gravity theories.
\end{conj}

\noindent The existence of a gravitational cosmic Galois group underlying topology change in \textit{all} quantum gravity theories fits very nicely with the conjuectural uniqueness of UV complete quantum gravity \cite{McNamara:2019rup}. We look forward to applying our setup to more complex models, and searching for a unified algebraic description of topology change in quantum gravity.

\section*{Acknowledgements}

We thank Kasia Budzik, Xi Dong, Kurusch Ebrahimi-Fard, Sergio Hernandez-Cuenca, Manki Kim, Henry Maxfield, Hirosi Ooguri, Gary Shiu, Cumrun Vafa, Wayne Weng and Xi Yin for helpful discussions. MM is supported by NSF grant DMS-2104330. JM is supported by the U.S. Department of Energy, Office of Science, Office of High Energy Physics, under Award Number DE-SC0011632.

\begin{appendix}
\section{Basic facts about Hopf algebras}\label{sec:Hopf_algs}

The main mathematical result in this text is that there is a Hopf algebra structure underlying the sum over topologies in the gravitational path integral. In this appendix, we review a few standard facts about Hopf algebras that have been useful in the main text.

\begin{defn} A \textit{bialgebra} $(\mathcal{H},\mu,\Delta,\eta,\varepsilon)$ over a field $k$ is a $k$-algebra $\mathcal{H}$ (with multiplication operation $\mu : \mathcal{H} \otimes \mathcal{H} \to \mathcal{H}$ and unit $\eta : k \to \mathcal{H}$) equipped with additional algebra homomorphisms $\Delta:\mathcal{H} \rightarrow \mathcal{H} \otimes \mathcal{H}$, called the \textit{coproduct}, and $\varepsilon: H\rightarrow k$, called the \textit{counit}, which are coassociative and counital, i.e., such that the following diagrams commute:
\[\begin{tikzcd}
	\mathcal{H} && {\mathcal{H}\otimes \mathcal{H}} \\
	{\mathcal{H}\otimes \mathcal{H}} && {\mathcal{H}\otimes \mathcal{H}\otimes \mathcal{H}} \\
	\mathcal{H} && {\mathcal{H}\otimes \mathcal{H}} \\
	{\mathcal{H}\otimes  \mathcal{H}} && \mathcal{H} \\
	\arrow["\Delta", from=1-1, to=1-3]
	\arrow["\Delta"', from=1-1, to=2-1]
	\arrow["{{\rm Id}\otimes \Delta}", from=1-3, to=2-3]
	\arrow["{\Delta\otimes {\rm Id}}"', from=2-1, to=2-3]
	\arrow["{\rm Id}", from=3-1, to=4-3]
	\arrow["\Delta", from=3-1, to=3-3]
	\arrow["{\varepsilon\otimes {\rm Id}}"', from=4-1, to=4-3]
	\arrow["\Delta"', from=3-1, to=4-1]
	\arrow["{\rm Id}\otimes\varepsilon", from=3-3, to=4-3]
\end{tikzcd}\]
Intuitively, one can think of the coproduct $\Delta$ as a decomposition operation, just as one may think of the product $\mu$ as a composition operation. In the case of Feynman diagrams, the coproduct \eqref{eq:CK_coprod} on the $\mathcal{H}_{\rm CK}$ extracts all possible sub-divergences; in the case of wormholes, the coproduct \eqref{eq:gravcoprod} on $\mathcal{H}_{\rm FdB}$ extracts all possible sub-wormholes.

A \textit{Hopf algebra} $(\mathcal{H},\mu,\Delta,\eta,\varepsilon)$ is a bialgebra equipped with an additional $k$-linear map $S : \mathcal{H} \to \mathcal{H}$ called the \textit{antipode} such that the following diagram commutes:
\[\begin{tikzcd}
	& {\mathcal{H}\otimes \mathcal{H}} && {\mathcal{H}\otimes \mathcal{H}} \\
	\mathcal{H} && k && \mathcal{H} \\
	& {\mathcal{H}\otimes \mathcal{H}} && {\mathcal{H}\otimes \mathcal{H}}
	\arrow["{S\otimes {\rm Id}}", from=1-2, to=1-4]
	\arrow["\mu", from=1-4, to=2-5]
	\arrow["\Delta", from=2-1, to=1-2]
	\arrow["\varepsilon"', from=2-1, to=2-3]
	\arrow["\Delta"', from=2-1, to=3-2]
	\arrow["\eta"', from=2-3, to=2-5]
	\arrow["{{\rm Id}\otimes S}"', from=3-2, to=3-4]
	\arrow["\mu"', from=3-4, to=2-5]
\end{tikzcd}\]
When it exists, the antipode is uniquely determined by the bialgebra structure on $\mathcal{H}$.
\end{defn}

In this paper, all Hopf algebras considered are commutative but not co-commutative. Under the standard duality between commutative rings and their spectra, commutative Hopf algebras are dual to groups:

\begin{prop}\label{prop:dual_group_to_Hopf_algebra}
Let $(\mathcal{H},\mu,\eta)$ be a commutative algebra over a field $k$. Then the structure of a coproduct $\Delta$, a counit $\varepsilon$, and an antipode $S$ making $\mathcal{H}$ into a Hopf algebra endows the spectrum $G = {\rm Spec}(\mathcal{H})$ with a group structure, called the dual group of $\mathcal{H}$. Under this equivalence, the unit $1 \in G$ is defined to be the counit $\varepsilon$, the group law $\star : G \times G \to G$ is defined in terms of the coproduct $\Delta$ by
\begin{equation}
(g_1 \star g_2)(X) = (g_1 \otimes g_2)(\Delta(X)),
\end{equation}
and the inverse map $(\cdot)^{-1} : G \to G$ is defined in terms of the antipode $S$ by
\begin{equation}
g^{-1}(X) = g(S(X)).
\end{equation}
In both of these formulas, $X$ represents a Hopf algebra element $x \in \mathcal{H}$, and group elements $g, g_1, g_2$ are defined to be algebra homomorphisms $\mathcal{H} \to k$.
\end{prop}

In the main text, $G_{\mathrm{CK}}$ and $G_{\mathrm{FdB}}$ are the dual groups of the Hopf algebras $\mathcal{H}_{\mathrm{CK}}$ and $\mathcal{H}_{\mathrm{FdB}}$ respectively. In both cases, it is essential that $G_{\mathrm{CK}}$ and $G_{\mathrm{FdB}}$ are \textit{groups}, in which every element is invertible, due to the existence of antipodes on $\mathcal{H}_{\mathrm{CK}}$ and $\mathcal{H}_{\mathrm{FdB}}$. In the case of perturbative renormalization, the antipode relates the divergences to the counterterms, and in the case of gravity the antipode relates the wormholes to the counter-wormholes. In both cases, the antipode can be defined recursively through one or the other version of the BPHZ procedure. Mathematically, such a recursive construction is possible whenever we have a \textit{connected graded} bialgebra:

\begin{prop}[\cite{Takeuchi1971FreeHA}]\label{prop:conn_graded}
Let $\mathcal{H}$ be a connected graded bialgebra over a field $k$: i.e., a bialgebra $\mathcal{H} = \bigoplus_{n \geq 0} \mathcal{H}_n$ such that $\mu, \Delta, \eta$ and $\varepsilon$ respect the grading, and such that $\mathcal{H}_0 = k$. Then $\mathcal{H}$ admits an antipode $S$, so it is a Hopf algebra.
\end{prop}
\begin{proof}
    From the axioms for $S$, we must have $S(1) = 1$. Now, for $n > 0$, let $X \in \mathcal{H}_n$ be a homogenous element of degree $n$. Noting that $\varepsilon$ is the projection onto $\mathcal{H}_0$ and applying counitality, we must have
    \begin{equation}
        \Delta(X) = X \otimes 1 + 1 \otimes X + \sum X' \otimes X'',
    \end{equation}
    where $X', X''$ are elements of strictly lower degree than $n$. We define $S(X)$ recursively by the formula
    \begin{equation}
        S(X) = - X - \sum S(X') X''.
    \end{equation}
    This formula defines an abstract version of the BPHZ procedure.
\end{proof}


\bibliographystyle{utphys}
\bibliography{ref}

\end{appendix}
\end{document}